\RequirePackage{fix-cm}
\documentclass[natbib]{svjour3arxiv}                     
\smartqed  
\usepackage{mathptmx}      
\usepackage[latin1]{inputenc}
\usepackage[english]{babel}
\usepackage{amsmath,amssymb}
\usepackage{mathtools}
\usepackage{graphicx}
\usepackage{mdframed,subfig}
\usepackage{booktabs,caption}
\usepackage{multicol,multirow}
\usepackage{makeidx}  
\usepackage{tikz}
\usetikzlibrary{positioning}
\def\fiber{\mathcal{F}_{N,t}}
\DeclareMathOperator\Poi{Poisson}
\DeclareMathOperator\Bin{Binomial}

\DeclareMathOperator\partition{part}
\DeclareMathOperator\variance{Var}
\DeclareMathOperator\bias{E}
\DeclareMathOperator\mad{MAD}

\begin{document}

\title{Markov Chain Monte Carlo sampling for conditional  tests
}
\subtitle{A link between permutation tests and algebraic statistics}


\author{Francesca Romana Crucinio         \and
        Roberto Fontana
}


\institute{Francesca Romana Crucinio \at
              Politecnico di Torino, Dipartimento di Scienze Matematiche \\
              \email{francesca.crucinio@gmail.com}
           \and
           Roberto Fontana \at
              Politecnico di Torino, Dipartimento di Scienze Matematiche \\
              Tel.: +39 011 0907504\\
              \email{roberto.fontana@polito.it}
}

\date{ }

\maketitle
\begin{abstract}
We consider conditional tests for non-negative discrete exponential families. We develop two Markov Chain Monte Carlo (MCMC) algorithms which allow us to sample from the conditional space and to perform approximated tests.
The first algorithm is based on the MCMC sampling described by \cite{sturmfels}. The second MCMC sampling consists in a more efficient algorithm which exploits the optimal partition of the conditional space into orbits of permutations. We thus establish a link between standard permutation and algebraic-statistics-based sampling. 
Through a simulation study we compare the exact cumulative distribution function (cdf) with the approximated cdfs which are obtained with the two MCMC samplings and the standard permutation sampling. 
We conclude that the MCMC sampling which exploits the partition of the conditional space into orbits of permutations gives an estimated cdf, under $H_0$, which is more reliable and converges to the exact cdf with the least steps. This sampling technique can also be used to build an approximation of the exact cdf when its exact computation is computationally infeasible.
\end{abstract}

\section{Introduction}
\label{sec:intro}
The problem of comparing two measures of location for two random samples is one of the classical problems which arise in statistics. 
We consider two independent samples, $\mathbf{Y}_1^{(n_1)}=(Y_1,\ldots,Y_{n_1})$ of size $n_1$ from $f_1$ and $\mathbf{Y}_2^{(n_2)}=(Y_{n_1+1},\ldots,Y_{n_1+n_2})$ of size $n_2$ from $f_2$, and we assume that $f_1$ and $f_2$ belong to the same non-negative discrete exponential family
\[f(y\vert \mu)=G(\mu)H(y)\exp\lbrace y\cdot\psi(\mu)\rbrace,\]
where $G$ only depends on the parameter $\mu$, $H$ is the underlying measure which only depends on data $y$ and $\psi$ is the natural parameter.

This assumption is not too restrictive. Most widely-used discrete distributions belong to the exponential family. Some examples are reported in Table \ref{tab:expfam}.

\begin{table}[h]
\centering
\caption{Examples of non-negative distributions of the discrete exponential family}
\label{tab:expfam}  
\begin{tabular}{l|llll}
\hline\noalign{\smallskip}
\textbf{Distribution} & $G(\mu)$&$H(y)$&$\psi(\mu)$&$\mathbb{E}[Y]$\\
\noalign{\smallskip}\hline\noalign{\smallskip}
Poisson&$e^{-\mu}$&$\frac{1}{y!}$&$\log(\mu)$&$\mu$\\
Geometric&$\mu$&$1$&$\log(1-\mu)$&$\frac{1}{\mu}$\\
Binomial ($k$ trials)&$k\log(1-\mu)$&$\binom{k}{y}$&$\log\left(\frac{\mu}{1-\mu}\right)$&$k\mu$\\
\hline\noalign{\smallskip}
\end{tabular}
\end{table}

We are interested in checking if the two distributions are equal, that is if the pooled sample $\mathbf{Y}=(\mathbf{Y}_1^{(n_{1})},\mathbf{Y}_2^{(n_{2})})$ is formed by $N=n_1+n_2$ observations coming from the same distribution $f=f_1=f_2$.
Thus we perform the hypothesis test
\begin{equation}
\label{eq:H0}
H_0: \mu_1=\mu_2\qquad \text{against} \qquad H_1: \mu_1\gtrless\mu_2.
\end{equation}

Several testing procedures are available, both parametric and non-parametric. We consider \emph{conditional} tests for discrete exponential families. To this aim we consider the joint distribution of sample $\mathbf{Y}$, which (with a slight abuse of notation) we denote by $f$
\begin{align*}
&f(y_1,\dots,y_{n_1},y_{n_1+1},\dots,y_{n_1+n_2})=\\
&=\prod_{i=1}^{n_1} f_1(y_i)\prod_{i=n_1+1}^{n_2+n_1} f_2(y_i)=\\
&=G(\mu_1)^{n_1}G(\mu_2)^{n_2}\cdot\prod_{i=1}^{n_1}H(y_i)\prod_{i=n_1+1}^{n_2+n_1}H(y_i)\cdot\\
&\exp\left\lbrace \left(\sum_{i=1}^{n_1} y_i\right)\left(\psi(\mu_1)- \psi(\mu_2)\right)+\left(\sum_{i=1}^{n_1+n_2}y_i\right)\psi(\mu_2)\right\rbrace.
\end{align*}

According to \cite{lehmann} there exists a uniformly most powerful unbiased (UMPU) test of
\begin{equation}
H_0:\psi(\mu_1)-\psi(\mu_2)=0\quad \text{against}\quad \psi(\mu_1)-\psi(\mu_2)\gtrless0
\label{eq:H0_umpu}
\end{equation}
performed conditionally on $T=U+V$, where $U=\sum_{i=1}^{n_1} Y_i$ and $V=\sum_{i=n_1+1}^{n_1+n_2} Y_i$.
We observe that the hypotheses in equation \eqref{eq:H0_umpu} are equivalent to the ones in \eqref{eq:H0}.

The conditioning on $T=\sum_{i=1}^{n_1+n_2} Y_i$ is justified by the fact that if we assume the standard one-way ANOVA model for the means of the two distributions, which according to \cite{glm} is $\psi(\mu_i)=\beta_0+\beta_1 x_i$ with $x_i=1$ if $1\leq i\leq n_1$ and $x_i=-1$ if $n_1+1\leq i\leq n_1+n_2$, the statistic $T=\sum_{i=1}^{n_1+n_2}Y_{i}$ is sufficient for the population constant $\beta_0$, which is the nuisance parameter of the test.
Hence we will consider tests which are performed considering only samples $\mathbf{Y}$ such that the sum $T=\sum_{i=1}^{n_1+n_2}Y_{i}$ of their elements is equal to the sum $t=\sum_{i=1}^{n_1+n_2}y_{i,obs}$ of the elements of the observed sample $\mathbf{y}_{obs}$.

Given that the rejection/non-rejection of the null hypothesis \eqref{eq:H0} depends on the p-value corresponding to the observed sample $\mathbf{y}_{obs}$, which is a function of the cumulative distribution of the test statistic under $H_0$, we will focus on efficient algorithms to compute a good estimate of the cumulative distribution function under $H_0$.

In Section \ref{sec:umpu} we briefly describe the exact uniformly most powerful unbiased conditional test presented in \cite{lehmann}. We will refer to this test as the UMPU test. In Section \ref{sec:mcmc} we describe a Markov Chain Monte Carlo algorithm which exploits Markov basis \citep{sturmfels} and the Metropolis-Hastings theorem \citep{casella} to approximate the distribution of the test statistic $U$ of the UMPU test. This procedure can be used when a computational-friendly form of the conditional distribution is not known.
In Section \ref{sec:permutation1} we prove that the set from which we sample with the MCMC procedure can be partitioned into orbits of permutations and that this partition is somehow \emph{optimal}, because we can exploit the inferentially equivalence of permutations to devise a 2-step sampling method which is more efficient in terms of convergence speed and reliability than the one presented in Section \ref{sec:mcmc}. This 2-step sampling consists in performing a MCMC sampling over the set of orbits and then a standard Monte Carlo sampling of the elements of the sampled orbit through which we build the usual Monte Carlo permutation cdf.
We will refer to this sampling method as the MCMC over the orbits.
In Section \ref{sec:compare} we compare the properties of the two estimators of the exact conditional cumulative distribution under $H_0$, the one described in Section \ref{sec:mcmc} and the Monte Carlo permutation cdf described in Section \ref{sec:permutation1}. 
In Section \ref{sec:permutation2} we briefly analyse the link between the MCMC sampling over the orbits and the standard permutation test. 
There follows a comparison through simulation study of the exact cumulative distribution function with the cdfs corresponding to the three sampling procedures (the two MCMCs and the standard permutation sampling) and a study of the convergence speed of the two MCMC routines (Section \ref{sec:sim}).
Finally, in Section \ref{sec:conclusion} we conclude that the 2-step sampling presented in Section \ref{sec:permutation1} is an efficient way to perform uniformly most powerful unbiased tests for any non-negative discrete exponential family. In addition we discuss some further extensions of this sampling algorithm to N-way ANOVA models.

\section{Exact Conditional Uniformly Most Powerful and Unbiased Test}
\label{sec:umpu}

Following \cite{lehmann} we are interested in the distribution of $U=\sum_{i=1}^{n_1} Y_i$ given $T=U+V$, where $V=\sum_{i=n_1+1}^{n_1+n_2} Y_i$.

When dealing with discrete exponential families the conditional distribution $U\vert T$ can be easily found. It is well known that if $U$ and $V$ are independent then
\begin{equation}
\label{eq:condition}
\mathbb{P}(U=u\vert T=t)=\frac{\mathbb{P}(U=u)\mathbb{P}(V=t-u)}{\mathbb{P}( T=t)}.
\end{equation}
Then, given that $U=\sum_{i=1}^{n_1} Y_i$ and $V=\sum_{i=n_1}^{n_1+n_2} Y_i$ are sum of i.i.d. random variables belonging to the same exponential family, their distribution is known and equation \eqref{eq:condition} becomes
\begin{multline}
\frac{G(\mu_1)^{n_1}G(\mu_2)^{n_2} e^{\left\lbrace u(\psi(\mu_1)-\psi(\mu_2))+t\psi(\mu_2)\right\rbrace}}
{G(\mu_1)^{n_1}G(\mu_2)^{n_2} e^{\left\lbrace t\psi(\mu_2)\right\rbrace}}\cdot\\
\frac{\underset{\mathbf{y}_1^{(n_1)}\in\mathcal{F}_{n_1,u}}\sum \underset{i=1}{\overset{n_1}\prod }H(y_i)\cdot\underset{\mathbf{y}_2^{(n_2)}\in\mathcal{F}_{n_2,t-u}}\sum\underset{i=n_1+1}{\overset{n_1+n_2}\prod }H(y_i)}{\overset{t}{\underset{u=0}\sum} e^{\left\lbrace u(\psi(\mu_1)-\psi(\mu_2))\right\rbrace} \underset{\mathbf{y}_1^{(n_1)}\in\mathcal{F}_{n_1,u}}\sum\underset{i=1}{\overset{n_1}\prod }H(y_i)\cdot\underset{\mathbf{y}_2^{(n_2)}\in\mathcal{F}_{n_2,t-u}}\sum\underset{i=n_1+1}{\overset{n_1+n_2}\prod }H(y_i)},
\end{multline}
where we denote by $\mathcal{F}_{n,x}$ the set of positive integer vectors of length $n$ with sum of entries equal to $x$.

Under $H_0:\psi(\mu_1)=\psi(\mu_2)$ and with some straightforward simplifications this reduces to
\begin{gather}
\label{eq:umpu}
\frac{\underset{\mathbf{y}_1^{(n_1)}\in\mathcal{F}_{n_1,u}}\sum\underset{i=1}{\overset{n_1}\prod }H(y_i)\cdot \underset{\mathbf{y}_2^{(n_2)}\in\mathcal{F}_{n_2,t-u}}\sum\underset{i=n_1+1}{\overset{n_1+n_2}\prod }H(y_i)}{\overset{t}{\underset{u=0}\sum}\underset{\mathbf{y}_1^{(n_1)}\in\mathcal{F}_{n_1,u}} \sum\underset{i=1}{\overset{n_1}\prod }H(y_i)\cdot\underset{\mathbf{y}_2^{(n_2)}\in\mathcal{F}_{n_2,t-u}}\sum\underset{i=n_1+1}{\overset{n_1+n_2}\prod }H(y_i)}.
\end{gather}

Now we can either find the critical values for any given risk of type I error or, alternatively, compute the p-value which corresponds to the observed value $u_{obs}$ of $U$, through the conditional cumulative distribution function $\mathbb{P}(U\leq u\vert T=t)$ under $H_0$. 

As an example consider two samples $Y_1,\ldots, Y_{n_1}$ distributed according to Poisson($\mu_1$) and $Y_{n_1+1},\ldots, Y_{n_1+n_2}$ distributed according to Poisson($\mu_2$). Under $H_0:\mu_1=\mu_2$, the distribution \eqref{eq:umpu} is a binomial distribution with probability of success $\theta_0=n_1/(n_1+n_2)$ and $t$ trials.

However computing the distribution \eqref{eq:umpu} is sometimes extremely difficult and time-consuming. In this case we can rely on the MCMC procedure described in Section \ref{sec:mcmc}.

Two examples of distributions for which a computational-friendly form of distribution \eqref{eq:umpu} is not easy to find are the log-series distribution (pag. 297 of \cite{kemp})
\[f(y\vert\mu)=-\frac{\mu^y}{y\log(1-\mu)}, \quad y>0,\ 0 < \mu <1,\]
and the lost-games distribution with fixed gambler's initial capital $j \geq 0$, fixed maximum achievable capital $a>0$ and $0.5<\mu<1$ (pag. 445 of \cite{kemp})
\[f(y\vert\mu)=\binom{2y+a-2j}{y-j}\frac{a}{2y+a-2j}\mu^{a+y-j}(1-\mu)^{y-j}.\]

\section{The Markov Chain Monte Carlo Test}
\label{sec:mcmc}

As a consequence of the conditioning on $T=\sum_{i=1}^{n_1+n_2} Y_i$, to perform an exact hypothesis test, the conditional space to be inspected under $H_0$, is the \emph{fiber} of vectors of size $N=n_1+n_2$ and with entries which add up to $t$
\begin{equation}
\label{eq:fiber}
\mathcal{F}_{N,t}=\lbrace(Y_1,\dots,Y_{n_1+n_2})\in \mathbb{N}^{n_1+n_2}: \sum_{i=1}^N Y_i= \mathbf{1}_N^T \mathbf{Y}=t\rbrace
\end{equation}
where $\mathbf{1}_N=\underbrace{(1, \ldots, 1)}_{N \text{ times}}$.

We consider the distribution of $U=\sum_{i=1}^{n_1} Y_i$ under $H_0$ over the fiber $\fiber$. We denote such conditional cumulative distribution function by $F_U(u\vert \fiber)$. We get
\begin{equation}
\label{eq:dist_fiber}
F_U(u\vert \fiber)=\mathbb{P}(U(\mathbf{y})\leq u\vert \mathbf{y}\in\fiber)=\sum_{\mathbf{y}\in\fiber}\mathbb{I}_{(U(\mathbf{y})\leq u)}(\mathbf{y})p(\mathbf{y}),
\end{equation}
where $U(\mathbf{y})=\sum_{i=1}^{n_1}y_i$ and $\mathbb{I}_{(U(\mathbf{y})\leq u)}(\mathbf{y})$ is $1$ if $U(\mathbf{y})\leq u$ and $0$ otherwise.

We focus on how well we can compute the conditional cumulative distribution function $F_U(u\vert \fiber)$ because p-values of the conditional tests are computed using this function. Specifically, following \cite{przy}, the p-value for the left one-sided test is $F_U(u_{obs}\vert \fiber)$, the p-value for the right one-sided test is $1-F_U(u_{obs}-1\vert \fiber)$ and for the two-sided test the p-value is $\min\lbrace 2 \cdot\min \lbrace \\F_U(u_{obs}\vert \fiber), 1-F_U(u_{obs}-1\vert \fiber) \rbrace ,1 \rbrace$.

We sample from the conditional space $\mathcal{F}_{N,t}$ defined in equation \eqref{eq:fiber} using a MCMC procedure. We know from \cite{sturmfels} that using a Markov basis we can explore the fiber $\fiber$ by adding and subtracting moves, obtaining a connected Markov chain. 

To find a Markov basis we consider the vector $\mathbf{1}_N$ which appears in the definition \eqref{eq:fiber} of $\fiber$ and depends on the sample size $N=n_1+n_2$ only. By means of the induction principle on $N$, it can be proved that a basis is always given by the $N\times (N-1)$ matrix
\begin{equation}
\label{eq:basis}
\mathcal{B}_{N}=\begin{bmatrix*}[r]
\mathbf{m}_1& \mathbf{m}_2&\ldots&\mathbf{m}_{N-1}
\end{bmatrix*}=\underbrace{\begin{bmatrix*}[r]
1&1&\dots &1\\
-1&0& \dots &0\\
0&-1&\dots&0\\
\vdots & \vdots & \vdots &\vdots\\
0&0&\dots &-1\\	
\end{bmatrix*}}_{N-1}
\left.
\begin{aligned}
  & \\
  & \\
  & \\
  & \\
  & \\
\end{aligned} 
\right\rbrace N,
\end{equation}
where each column represents a move. The same basis can be found using the \texttt{4ti2} software \citep{4ti2}. 

Thanks to this basis we can build a graph $\mathcal{G}=(V,E)$ over the fiber $\mathcal{F}_{N,t}$. The set of vertices $V$ is given by the vectors $\mathbf{y}\in\fiber$ and each pair of vectors $\mathbf{y},\ \mathbf{x}\in \mathcal{F}_{N,t}$ is linked by an edge if a move $\mathbf{m}\in\mathcal{B}_{N}$ exists such that $\mathbf{y}=\mathbf{x}\pm\mathbf{m}$.
The number of vertices is $\vert V\vert=\binom{t+N-1}{N-1}$ and the number of edges is given by
\begin{align}
\label{eq:edges}
\vert E\vert =&(N-1)\binom{t-1}{N-1}+\notag\\
&\frac{1}{2}\sum_{z=1}^{N-1}(2N-2-z)\binom{t-1}{N-1-z} \binom{N-1}{z}+\notag\\
&\frac{1}{2}\sum_{z=1}^{N-1}(N-z)\binom{t-1}{N-1-z}\binom{N-1}{z-1}.
\end{align}
Furthermore $\mathcal{G}$ is connected (by definition of Markov basis) and bipartite. We invite the reader to refer to Appendix \ref{app:graph} for further details on the characteristics of the graph $\mathcal{G}$ (including the proof of equation \eqref{eq:edges}).
\begin{figure*}
\begin{center}

\resizebox{\textwidth}{!}{%
\begin{tikzpicture}[-latex,auto,node distance=2cm and 2cm,on grid, 
	semithick, state/.style={circle,top color=white,bottom color=white, 
  	draw, black, text=black,minimum width=1cm}] 
%
%
  	
\node[state] (AA) {$600$}; 
\node[state] (AB) [right=of AA] {$510$};
\node[state] (AC) [right=of AB] {$420$}; 
\node[state] (AD) [right=of AC] {$330$}; 
\node[state] (AE) [right=of AD] {$240$}; 
\node[state] (AF) [right=of AE] {$150$}; 
\node[state] (AG) [right=of AF] {$060$}; 
\node        (6) [right=of AG] {\large $u=6$};

\node[state] (BA) [below=of AA] {$501$};
\node[state] (BB) [right=of BA] {$411$};
\node[state] (BC) [right=of BB] {$321$};
\node[state] (BD) [right=of BC] {$231$};
\node[state] (BE) [right=of BD] {$141$};
\node[state] (BF) [right=of BE] {$051$};
\node        (5)  [below=of 6] {\large $u=5$};

\node[state] (CA) [below=of BA] {$402$};
\node[state] (CB) [right=of CA] {$312$};
\node[state] (CC) [right=of CB] {$222$};
\node[state] (CD) [right=of CC] {$132$};
\node[state] (CE) [right=of CD] {$042$};
\node        (4)  [below=of 5] {\large $u=4$};

\node[state] (DA) [below=of CA] {$303$};
\node[state] (DB) [right=of DA] {$213$};
\node[state] (DC) [right=of DB] {$123$};
\node[state] (DD) [right=of DC] {$033$};
\node        (3)  [below=of 4] {\large $u=3$};

\node[state] (EA) [below=of DA] {$204$};
\node[state] (EB) [right=of EA] {$114$};
\node[state] (EC) [right=of EB] {$024$};
\node        (2)  [below=of 3] {\large $u=2$};

\node[state] (FA) [below=of EA] {$105$};
\node[state] (FB) [right=of FA] {$015$};
\node        (1)  [below=of 2] {\large $u=1$};

\node[state] (GA) [below=of FA] {$006$};
\node        (0)  [below=of 1] {\large $u=0$};

\path[-] (AA) edge node[above]{$\mathbf{m}_1$} (AB);
\path[-] (AB) edge node[above]{$\mathbf{m}_1$} (AC);
\path[-] (AC) edge node[above]{$\mathbf{m}_1$} (AD);
\path[-] (AD) edge node[above]{$\mathbf{m}_1$} (AE);
\path[-] (AE) edge node[above]{$\mathbf{m}_1$} (AF);
\path[-] (AF) edge node[above]{$\mathbf{m}_1$} (AG);
\path[-] (BA) edge node[above]{$\mathbf{m}_1$} (BB);
\path[-] (BB) edge node[above]{$\mathbf{m}_1$} (BC);
\path[-] (BC) edge node[above]{$\mathbf{m}_1$} (BD);
\path[-] (BD) edge node[above]{$\mathbf{m}_1$} (BE);
\path[-] (BE) edge node[above]{$\mathbf{m}_1$} (BF);
\path[-] (CA) edge node[above]{$\mathbf{m}_1$} (CB);
\path[-] (CB) edge node[above]{$\mathbf{m}_1$} (CC);
\path[-] (CC) edge node[above]{$\mathbf{m}_1$} (CD);
\path[-] (CD) edge node[above]{$\mathbf{m}_1$} (CE);
\path[-] (DA) edge node[above]{$\mathbf{m}_1$} (DB);
\path[-] (DB) edge node[above]{$\mathbf{m}_1$} (DC);
\path[-] (DC) edge node[above]{$\mathbf{m}_1$} (DD);
\path[-] (EA) edge node[above]{$\mathbf{m}_1$} (EB);
\path[-] (EB) edge node[above]{$\mathbf{m}_1$} (EC);
\path[-] (FA) edge node[above]{$\mathbf{m}_1$} (FB);

\path[-] (AA) edge node[left]{$\mathbf{m}_2$} (BA);
\path[-] (BA) edge node[left]{$\mathbf{m}_2$} (CA);
\path[-] (CA) edge node[left]{$\mathbf{m}_2$} (DA);
\path[-] (DA) edge node[left]{$\mathbf{m}_2$} (EA);
\path[-] (EA) edge node[left]{$\mathbf{m}_2$} (FA);
\path[-] (FA) edge node[left]{$\mathbf{m}_2$} (GA);
\path[-] (AB) edge node[left]{$\mathbf{m}_2$} (BB);
\path[-] (BB) edge node[left]{$\mathbf{m}_2$} (CB);
\path[-] (CB) edge node[left]{$\mathbf{m}_2$} (DB);
\path[-] (DB) edge node[left]{$\mathbf{m}_2$} (EB);
\path[-] (EB) edge node[left]{$\mathbf{m}_2$} (FB);
\path[-] (AC) edge node[left]{$\mathbf{m}_2$} (BC);
\path[-] (BC) edge node[left]{$\mathbf{m}_2$} (CC);
\path[-] (CC) edge node[left]{$\mathbf{m}_2$} (DC);
\path[-] (DC) edge node[left]{$\mathbf{m}_2$} (EC);
\path[-] (AD) edge node[left]{$\mathbf{m}_2$} (BD);
\path[-] (BD) edge node[left]{$\mathbf{m}_2$} (CD);
\path[-] (CD) edge node[left]{$\mathbf{m}_2$} (DD);
\path[-] (AE) edge node[left]{$\mathbf{m}_2$} (BE);
\path[-] (BE) edge node[left]{$\mathbf{m}_2$} (CE);
\path[-] (AF) edge node[left]{$\mathbf{m}_2$} (BF);

\end{tikzpicture}
}
\end{center} 
\caption{Graph on the fiber $\mathcal{F}_{N,t}$ with $t=6$ and $N=3$}
\label{fig:grafo}
\end{figure*}
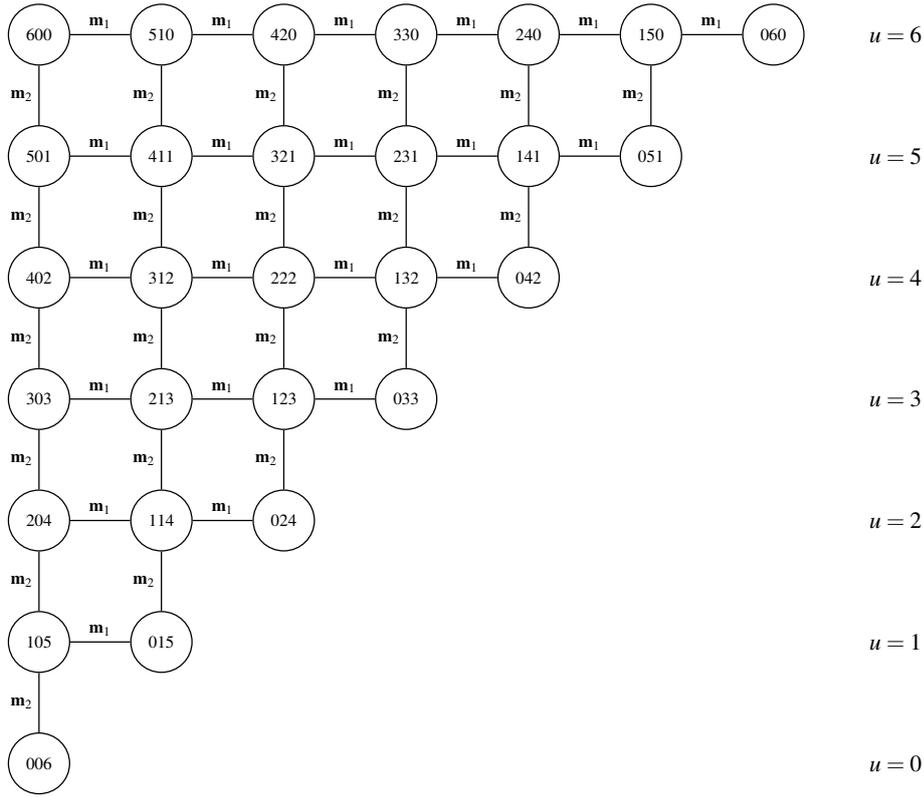

We consider $n_1=2$, $n_2=1$ and $t=6$ as an example. The cardinality of the fiber $\mathcal{F}_{3,6}$ is $\binom{6+3-1}{3-1}=28$ and the corresponding Markov basis is
\begin{equation} \label{eq:b3}
\mathcal{B}_{3}=\begin{bmatrix*}[r]
\mathbf{m}_1&\mathbf{m}_2\\
\end{bmatrix*}=\begin{bmatrix*}[r]
1&1\\
-1&0\\
0&-1\\
\end{bmatrix*}.
\end{equation}
This results in the graph $\mathcal{G}$ shown in Figure \ref{fig:grafo}. The graph has 28 vertices and 42 edges, as given by previous formulae.

Thanks to the Markov basis $\mathcal{B}_{N}$ we set up a connected Markov chain over $\mathcal{F}_{N,t}$ and, under $H_0:\psi(\mu_1)=\psi(\mu_2)=\psi(\mu)$ (which is equivalent to $H_0:\mu_1=\mu_2=\mu$), we use the Metropolis-Hastings algorithm (as in \cite{aokitakemurabook} and \cite{aokitakemura}) to modify transition probabilities and grant convergence to
\begin{align}
\label{eq:target}
p(\mathbf{y})\equiv f(y\vert\mu) &=\prod_{i=1}^N G(\mu)H(y_i)\exp\lbrace y_i\psi(\mu) \rbrace = \\ \notag
&=G(\mu)^N \exp\lbrace \psi(\mu)t \rbrace\prod_{i=1}^N H(y_i)= \\ \notag
&=C \prod_{i=1}^N H(y_i),
\end{align}
where $C=G(\mu)^N \exp\lbrace \psi(\mu)t \rbrace$.


We set the initial state $\mathbf{y}=\mathbf{y}_{obs}$, where $\mathbf{y}_{obs}$ is the observed sample, and at every step we select a random move $\mathbf{m}_K$, $1\leq K\leq N-1$, from $\mathcal{B}_{N}$ and a random sign $\varepsilon=\pm1$. If $\mathbf{y}+\varepsilon\cdot\mathbf{m}_K$ has not negative entries (i.e. $\varepsilon\cdot\mathbf{m}_K$ is admissible) the transition occurs with probability
\begin{align}
q&=\min\left\lbrace 1,\frac{p(\mathbf{y}+\varepsilon\cdot\mathbf{m}_K)}{p(\mathbf{y})}\right\rbrace=\notag \\
&=\min\left\lbrace 1,\frac{C\cdot\prod_{i=1}^N H((\mathbf{y}+\varepsilon\cdot\mathbf{m}_K)_i)}{C\cdot\prod_{i=1}^N H(y_i)}\right\rbrace= \notag\\
&=\min\left\lbrace 1,\frac{H(y_1+\varepsilon)\cdot H(y_{K+1}-\varepsilon)}{H(y_1)\cdot H(y_{K+1})}\right\rbrace.
\label{eq:trans_q}
\end{align}

From equation \eqref{eq:trans_q} we observe that $C$ is not needed in order to perform the Metropolis-Hastings algorithm. 
However in Section \ref{sec:permutation1} we will give a simple way to compute the normalizing constant $C$. 

Through the walk on the fiber we sample from $\fiber$ and get an approximation of the cumulative distribution function of $U$ given $T$ shown in equation \eqref{eq:dist_fiber}.

There exists an accelerated version of this algorithm, proposed by \cite{sturmfels} and frequently adopted in the literature (see \cite{aokitakemurabook} and \cite{aokitakemura}), since it grants a much faster convergence to the target distribution.
At each step, given $\mathbf{y}\in\fiber$, we randomly select a move $\mathbf{m}_K\in\mathcal{B}_N$ and we consider all possible transitions $\mathbf{y}+\lambda\cdot \mathbf{m}_K$ such that $\mathbf{y}+\lambda\cdot \mathbf{m}_K \in \fiber$ with $\lambda \in \mathbb{Z}$. The set of multipliers $\lambda$ which are admissible is given by
$$L=\left\lbrace \lambda\in\mathbb{Z}: \mathbf{y}+\lambda\cdot \mathbf{m}_K\geq0\right\rbrace.$$

\noindent
It can be shown that, in this case,
$$
L=\left[-y_1,y_{K+1}\right]\cap \mathbb{Z}.
$$
The transition probabilities $\forall \lambda^*\in L$ are
\begin{align}
\label{eq:trans_q_lambda}
q_{\lambda^*}&=\frac{C\cdot\prod_{i=1}^N H((\mathbf{y}+\lambda^*\mathbf{m}_K)_i)}{\underset{\lambda\in L}\sum C\cdot\prod_{i=1}^N H((\mathbf{y}+\lambda\mathbf{m}_K)_i)}=\notag \\
&=\frac{H(y_1+\lambda^*)\cdot H(y_{K+1}-\lambda^*)}{\underset{\lambda\in L}\sum H(y_1+\lambda)\cdot H(y_{K+1}-\lambda)}.
\end{align}

As we observed for the standard MCMC algorithm, also in the accelerated case there is no need to compute the normalizing constant $C$. 

\subsection{The MCMC algorithm} \label{algo:mcmc}
In this work we use the accelerated version of the MCMC algorithm. 
The algorithm needs 5 inputs: the observed vector $\mathbf{y}_{obs}$, the Markov basis $\mathcal{B}_{N}$, the test statistic $U(\cdot)$, the underlying measure $H(\cdot)$ and the total number of simulation $N_{sim}$. The algorithm provides an estimate of the distribution of $U$ given $T$ under $H_0$, $F_U(u\vert \fiber)$, as output.

\begin{description}
\item[\textbf{Step 1}]: Set $u_{obs}=U(\mathbf{y}_{obs})$, $\mathbf{y}=\mathbf{y}_{obs}$ and a counter $i=0$.
\item[\textbf{Step 2}]: Draw $\mathbf{m}_K$, $1\leq K\leq N-1$, randomly from $\mathcal{B}_{N}$.
\item[\textbf{Step 3}]: Select $\lambda^*$ from $L=\left[-y_1,y_{K+1}\right]\cap \mathbb{Z}$ with probability
\[q_{\lambda^*}=\frac{H(y_1+\lambda^*)\cdot H(y_{K+1}-\lambda^*)}{\underset{\lambda\in L}\sum H(y_1+\lambda)\cdot H(y_{K+1}-\lambda)}.\]
Update $\mathbf{y}=\mathbf{y}+\lambda^*\cdot\mathbf{m}_K$.
\item[\textbf{Step 4}]: Compute $U(\mathbf{y})$. If $U(\mathbf{y})\leq u_{obs}$ then $i=i+1$.
\item[\textbf{Step 5}]: Repeat steps 2 to 4 $N_{sim}$ times.
\item[\textbf{Step 6}]: Compute an estimate of the distribution of $U$ under $H_0$ as $i/N_{sim}$.
\end{description}

\section{Partition of the Fiber with Orbits of Permutations}
\label{sec:permutation1}
In this section we present the link between \emph{orbits of permutations} and the fiber $\fiber$. We will exploit this link to build a much more efficient MCMC procedure to sample from $\fiber$.

Consider any permutation $\mathbf{y}^*$ of a given vector $\mathbf{y}$. Clearly $\mathbf{y}^*\in\mathcal{F}_{N,t}$ because the operation of permuting does not change the sum of the entries.
Therefore every permutation can be written as the sum of $\mathbf{y}$ with an appropriate sequence of moves in $\mathcal{B}_{N}$.
As a consequence, for every vector $\mathbf{y}\in\mathcal{F}_{N,t}$, the \emph{orbit of permutations} which contains $\mathbf{y}$ (that we denote by $\pi_{\mathbf{y}}$) is a subset of the fiber $\mathcal{F}_{N,t}$. Given that different orbits do not intersect, the fiber $\mathcal{F}_{N,t}$ can be partitioned into orbits of permutations.

The number of orbits is given by the number of partitions of $t$ into $N$ or fewer parts. This number is $\partition(t,N)$ (where $\partition$ is the partition function defined in \cite{partitions} and \cite{partitions2}). The values of the partition function can be computed using the recurrence $\partition(t,N)=\partition(t,N-1)+\partition(t-N,N)$ and depend on both the sample size $N$ and the sum of entries $t$.

An interesting property of the partition into orbits of permutations is that the elements which belong to the same orbit have the same probability of being sampled from the set $\mathcal{F}_{N,t}$ (this property is known as inferentially equivalence of permutations, see \cite{pesarin2010permutation} as a recent reference on this subject):
\begin{equation}
\label{eq:inf_equi}
p(\mathbf{y})=C\prod_{i=1}^N H(y_i)=C\prod_{i=1}^N H(y_{r_i})=C\prod_{i=1}^N H(y^*_i)=p(\mathbf{y}^*),
\end{equation}
where $(r_1,\dots,r_N)$ is any permutation of $(1,\dots,N)$ and $C=G(\mu)^N \exp\lbrace \psi(\mu)t \rbrace$.

Thanks to this property we can compute the conditional probability $\mathbb{P}(\mathbf{y}\vert \pi_{\mathbf{y}})$ of choosing an observation $\mathbf{y}$ inside an orbit $\pi_{\mathbf{y}}$. It holds that 
\begin{equation}
\label{eq:p_orbita}
\mathbb{P}(\pi_{\mathbf{y}})=
\mathbb{P}(\bigcup\limits_{\mathbf{y}^*\in\pi_{\mathbf{y}}} \left\{\mathbf{y}^{*}\right\})
= \sum_{\mathbf{y}^*\in\pi_{\mathbf{y}}}p(\mathbf{y}^*) 
=\#\pi_{\mathbf{y}}\cdot p(\mathbf{y})=\#\pi_{\mathbf{y}}\cdot C\prod_{i=1}^N H(y_i),
\end{equation}

where $\#\pi_{\mathbf{y}}$ is the cardinality of $\pi_{\mathbf{y}}$, i.e. the number of distinct permutations in $\pi_{\mathbf{y}}$ (see equation \eqref{eq:exact_card_orbit} for a formula to compute such a quantity). We denote by $p_\pi(\pi_{\mathbf{y}})$ the probability $\mathbb{P}(\pi_{\mathbf{y}})$.

By conditioning, $p(\mathbf{y})$ can be written as $p_\pi(\pi_{\mathbf{y}})\cdot \mathbb{P}(\mathbf{y}\vert \pi_{\mathbf{y}}) $  and given that from \eqref{eq:p_orbita} we get $p(\mathbf{y})=p_\pi(\pi_{\mathbf{y}})/\#\pi_{\mathbf{y}}$, we conclude that
\begin{equation}
\label{eq:condpi}
\mathbb{P}(\mathbf{y}\vert \pi_{\mathbf{y}})=1/\#\pi_{\mathbf{y}},
\end{equation}
i.e. given an orbit, the probability of selecting one of its elements is uniform.

It is worth noting that this is not true for every partition $\Gamma=\lbrace \gamma\rbrace$ of $\mathcal{F}_{N,t}$ because $\mathbb{P}(\mathbf{y}\vert\gamma)$ is not constant over $\gamma$ in general. 

We observe that although the normalizing constant $C$ defined in \eqref{eq:target} is not needed to perform the MCMC samplings described in Section \ref{sec:mcmc}, this can be easily computed by means of \eqref{eq:p_orbita} without knowing the parameter $\mu=\mu_1=\mu_2$:
\begin{align*}
1&=\underset{\pi_{\mathbf{y}}\subseteq\fiber}\sum p_\pi(\pi_{\mathbf{y}})=\underset{\pi_{\mathbf{y}}\subseteq\fiber}\sum \left(\#\pi_{\mathbf{y}}\cdot C\prod_{i=1}^N H(y_i)\right)=\\
&=C\cdot\underset{\pi_{\mathbf{y}}\subseteq\fiber}\sum \left(\#\pi_{\mathbf{y}}\cdot \prod_{i=1}^N H(y_i)\right).
\end{align*}
Hence $C=\left(\underset{\pi_{\mathbf{y}}\subseteq\fiber}\sum \#\pi_{\mathbf{y}}\cdot \prod_{i=1}^N H(y_i)\right)^{-1}$, an expression that does not contain the unknown parameter $\mu=\mu_1=\mu_2$.

As an example let us consider again the fiber $\mathcal{F}_{3,6}$ in Figure \ref{fig:grafo}. It can be partitioned into $\partition(6,3)=7$ orbits. If we assume that the distribution of data under $H_0$ is Poisson($\mu$), we get $H(y_i)=1/y_i!$ and $C=80/81$. We can then compute the probability $p_\pi(\pi_{\mathbf{y}})$ of each orbit given the probability $p(\mathbf{y})$ of one of its generators. The results are shown in Table \ref{tab:7orbits}.

\begin{table}[ht]
\centering
\caption{Probabilities of the 7 orbits contained in $\mathcal{F}_{3,6}$}
\label{tab:7orbits}
\begin{tabular}{cccr}
\hline\noalign{\smallskip}
$\mathbf{y}$&$p(\mathbf{y})$&$\#\pi_{\mathbf{y}}$&$ p_{\pi}(\pi_{\mathbf{y}})$\\
\noalign{\smallskip}\hline\noalign{\smallskip}
$(6,0,0)$&$\frac{80}{81}\cdot\frac{1}{  6!0!0!}$&3&$3/729$\\
\noalign{\smallskip}
$(5,1,0)$&$\frac{80}{81}\cdot\frac{1}{  5!1!0!}$&6&$36/729$\\
\noalign{\smallskip}
$(4,2,0)$&$\frac{80}{81}\cdot\frac{1}{ 4!2!0!}$&6&$90/729$\\
\noalign{\smallskip}
$(3,3,0)$&$\frac{80}{81}\cdot\frac{1}{ 3!3!0!}$&3&$60/729$\\
\noalign{\smallskip}
$(3,2,1)$&$\frac{80}{81}\cdot\frac{1}{  3!2!1!}$&6&$360/729$\\
\noalign{\smallskip}
$(4,1,1)$&$\frac{80}{81}\cdot\frac{1}{ 4!1!1!}$&3&$90/729$\\
\noalign{\smallskip}
$(2,2,2)$&$\frac{80}{81}\cdot\frac{1}{ 2!2!2!}$&1&$90/729$\\
\noalign{\smallskip}\hline
\end{tabular}
\end{table}

Given that the vectors which belong to the same orbit have the same probability, if the cardinality of the orbits is too large, we can use a standard Monte Carlo to sample from each orbit.

We devise a 2-step sampling method. We sample an orbit $\pi$ from the set of orbits $\pi\subseteq \fiber$ with probability  $p_\pi(\pi)$ and then we use a standard Monte Carlo to sample vectors $\mathbf{y}\in\pi$ with uniform distribution.

This procedure is more efficient than the one presented in Section \ref{sec:mcmc} because first we sample from the set of orbits, whose cardinality $\partition(t,N)$ is much smaller than the cardinality of the fiber $\binom{t+N-1}{N-1}$,  then we sample from a given orbit using the uniform distribution, which is extremely fast from a computational point of view.

As an example refer to the fiber $\mathcal{F}_{3,6}$ mentioned above, for which the number of orbits is $\partition(t,N)=7$ and the number of elements in the fiber is $\binom{t+N-1}{N-1}=28$. We can sample orbits $\pi$ with the probabilities $p_\pi(\pi)$ in Table \ref{tab:7orbits} by means of a MCMC routine, and then perform a standard Monte Carlo with uniform probabilities $1/\#\pi$ over every sampled orbit $\pi$.

In order to sample from the set of orbits of permutations we write the vectors $\mathbf{y}=(y_i:\ i=1,\ldots,N)\in\fiber$ in terms of the frequencies of each value of the $y_i$. The admissible values for each $y_i$ are $0\leq y_i\leq t$ and we associate the frequency $f_j$ to each $0\leq j\leq t$. As an example consider the vector $(0,4,2)\in\mathcal{F}_{3,6}$, whose corresponding frequency vector $\mathbf{f}=(f_j:\ j=0,\ldots,t)$ is shown in Table \ref{tab:frequenza}. Thanks to this notation vectors $(0,4,2)$ and $(2,0,4)$, which belong to the same orbit, correspond to the same frequency vector. We represent an orbit $\pi$ by the corresponding frequency vector $\mathbf{f}_{\pi}$.
\begin{table}[h]
\centering
\caption{Frequency-based representation of vector $(0,4,2)$}
\label{tab:frequenza}    
\begin{tabular}{llllllll}
\hline\noalign{\smallskip}
$j$ & 0&1	&2		&3	&4	&5	&6\\
\hline\noalign{\smallskip}
$f_j$ &1	&0	&1	&0	&1	&0	&0\\
\hline\noalign{\smallskip}
\end{tabular}
\end{table}

We denote by $\mathcal{F}^{(\pi)}_{N,t}$ the set of orbits of permutations $\pi$ contained in $\fiber$:
\begin{equation}
\label{eq:orbit_fiber}
\mathcal{F}^{(\pi)}_{N,t}=\left\lbrace (f_0,\ldots,f_t)\in\mathbb{N}^{t+1}: \sum_{j=0}^t jf_j=t\ \text{and}\  \sum_{j=0}^t f_j=N\right\rbrace;
\end{equation}
or in matrix notation
\begin{equation}
\label{eq:orbit_fiber_matrix}
\mathcal{F}^{(\pi)}_{N,t}=\left\lbrace \mathbf{f}\in\mathbb{N}^{t+1}: A^{(\pi)}_t \mathbf{f}=\begin{pmatrix}
t\\
N
\end{pmatrix}\right\rbrace,
\end{equation}
where $A^{(\pi)}_t=\begin{pmatrix}
0&1&\ldots&t\\
1&1&\ldots&1
\end{pmatrix}$. We observe that the matrix $A^{(\pi)}_t$ only depends on $t$, the sum of entries of $\mathbf{y}\in\fiber$, and not on the sample size $N$.

Thanks to this frequency-based notation and by means of equation \eqref{eq:p_orbita} we can compute the probability of a generic orbit $\pi$:
\begin{equation}
\label{eq:exact_p_orbit}
p_\pi(\pi) \equiv p_\pi(\mathbf{f}_\pi)=\#\pi_{\mathbf{y}}\cdot C\prod_{i=1}^N H(y_i)= \#\pi_{\mathbf{y}}\cdot C\prod_{j=0}^t H(j)^{f_j}
\end{equation}
where
\begin{equation}
\label{eq:exact_card_orbit}
\#\pi_{\mathbf{y}}=\frac{N!}{f_0!\cdot\ldots\cdot f_t!}.
\end{equation}
Hence the orbits $\pi \equiv \mathbf{f}_{\pi}$ follow a multinomial distribution constrained to $\sum_{j=0}^t j f_j=t$.

To perform the MCMC sampling over the orbits of permutations we need a Markov basis $\mathcal{B}^{(\pi)}_t$ for the fiber $\mathcal{F}^{(\pi)}_{N,t}$. We show how to build a Markov basis $\mathcal{B}^{(\pi)}_t$ for the fiber $\mathcal{F}^{(\pi)}_{N,t}$ in the Appendix \ref{app:mb}. We observe that $\mathcal{B}^{(\pi)}_t$ is a subset of the Markov basis given by the software \texttt{4ti2} \citep{4ti2} when the matrix $A^{(\pi)}_t$ is given as input. This is due to the fact that \texttt{4ti2} gives a basis for \emph{all} the fibers $\lbrace \mathbf{f} :A^{(\pi)}_t \mathbf{f}=\mathbf{b}, \mathbf{b}\geq \mathbf{0}\rbrace$, 
while we have built a Markov basis for the specific fiber $\mathcal{F}^{(\pi)}_{N,t}$, where $\mathbf{b}=\begin{pmatrix}
t\\
N
\end{pmatrix}$. It happens that some moves provided by  \texttt{4ti2} cannot be applied  because the entries of the resulting vector become negative. The availability of a specific Markov basis is an advantage from a computational point of view because no computational time is wasted in considering a-priori non-admissible moves. 

In the Appendix \ref{app:mb} we also show how to compute the number of moves in $\mathcal{B}^{(\pi)}_t$.

As an example we report the basis for $\mathcal{F}_{N,6}^{(\pi)}$. There are $t^2/4=9$ moves, each one represented as a column of the following matrix:
\begin{align} \label{eq:b6}
\mathcal{B}^{(\pi)}_6&=\begin{bmatrix*}[r]
\mathbf{m}^{(\pi)}_1\ldots \mathbf{m}^{(\pi)}_9
\end{bmatrix*}=\notag \\
&=\begin{bmatrix*}[r]
-1 &-1& -1& -1& -1& -1& -1& -1& -1\\ 
2 &1& 1& 0& 1& 0& 1& 0& 0 \\
-1 &1 &0& 2& 0& 1& 0& 1& 0 \\
0 &-1& 1& 0& 0& 1& 0& 0& 2 \\
0 &0& -1& -1& 1& 0& 0 &1 &0 \\
0 &0& 0& 0& -1& -1& 1 &0& 0 \\
0 &0 &0 &0 &0& 0 &-1 &-1 &-1 \\
\end{bmatrix*}.
\end{align}
%

Finally, we set up a connected Markov chain over $\mathcal{F}^{(\pi)}_{N,t}$, under $H_0$. We do not consider the accelerated version of the MCMC procedure, but only the standard one, because as we will see in Section \ref{sec:sim} the convergence speed of the standard algorithm is already very high.

As we did in Section \ref{sec:mcmc}, we modify transition probabilities according to the Metropolis-Hastings theorem to grant convergence to $p_\pi(\pi)$ as defined in equation \eqref{eq:exact_p_orbit}. If we are in orbit $\pi=\mathbf{f}_\pi$ and $\varepsilon\cdot \mathbf{m}^{(\pi)}$, with $\varepsilon=\pm1$, is admissible for $\mathbf{f}_\pi=(f_0,\ldots,f_t)$, the transition toward $\mathbf{f}_\pi^{'}=(f'_0,\ldots,f'_t)=\mathbf{f}_\pi+\varepsilon\cdot \mathbf{m}^{(\pi)}$ occurs with probability
\begin{align*}
q_{\pi}&=\min\left\lbrace 1, \frac{p_\pi(\mathbf{f}_\pi^{'})}{p_\pi(\mathbf{f}_\pi)}\right\rbrace=\\
&=\min\left\lbrace 1, \frac{\frac{N!}{f_0^{'}!\cdot\ldots \cdot f_t^{'}!}\cdot C\cdot\prod_{j=0}^t H(j)^{f_j^{'}}}
{\frac{N!}{f_0!\cdot\ldots \cdot f_t!}\cdot C\cdot\prod_{j=0}^t H(j)^{f_j}}\right\rbrace=\\
&=\min\left\lbrace 1, \frac{f_0!\cdot\ldots \cdot f_t!\cdot\prod_{j=0}^t H(j)^{f_j^{'}}}{
f_0^{'}!\cdot\ldots \cdot f_t^{'}!\cdot\prod_{j=0}^t H(j)^{f_j}}\right\rbrace.
\end{align*}
As we observed for the others MCMC routines, also in this case there is no need to know the normalizing constant $C$.

Thanks to this analysis we propose a new MCMC sampling algorithm. This new algorithm makes use of the distribution of $U$ given $T$ over one orbit $\pi$, that we denote by $F_U(u\vert \pi)$,
\begin{equation}
\label{eq:dist_orbit}
F_U(u\vert \pi)=\mathbb{P}(U(\mathbf{y})\leq u\vert \mathbf{y}\in\pi)=\sum_{\mathbf{y}\in \pi}\mathbb{P}(\mathbf{y}\vert \pi)\mathbb{I}_{(U(\mathbf{y})\leq u)}(\mathbf{y})=\frac{1}{\# \pi}\sum_{\mathbf{y}\in \pi}\mathbb{I}_{(U(\mathbf{y})\leq u)}(\mathbf{y}).
\end{equation}

\subsection{The orbit-based MCMC algorithm} \label{algo:mcmc_orbit} 
The algorithm needs 5 inputs: the observed vector $\mathbf{y}_{obs}$, the Markov basis $\mathcal{B}^{(\pi)}_t$, the test statistic $U(\cdot)$, the probability $p_\pi(\cdot)$ and the total number of simulation $N_{sim}$. The algorithm gives an estimate of  the distribution of $U$ given $T$ under $H_0$, $F_U(u\vert \fiber)$, as output.
\begin{description}
\item[\textbf{Step 1}]: Compute $u_{obs}=U(\mathbf{y}_{obs})$ and perform a standard Monte Carlo over $\pi_{\mathbf{y}_{obs}}$ to get $\hat{F}_U(u_{obs}\vert \pi_{\mathbf{y}_{obs}})$, the Monte Carlo approximation of the cdf $F_U(u_{obs}\vert \pi_{\mathbf{y}_{obs}})$.
\item[\textbf{Step 2}]: Transform $\mathbf{y}_{obs}$ into its frequency-based representation $\mathbf{f}_{\pi,obs}$ and set $\mathbf{f}_{\pi}=\mathbf{f}_{\pi,obs}$.
\item[\textbf{Step 3}]: Select $\mathbf{m}_K^{(\pi)}$ randomly from $\mathcal{B}^{(\pi)}_t$ and a random sign $\varepsilon$.
\item[\textbf{Step 4}]: If $\mathbf{f}_{\pi}+\varepsilon\cdot\mathbf{m}_K^{(\pi)}\geq0$, then update $\mathbf{f}_{\pi}$ with probability
\[
\min\left\lbrace 1, \frac{f_0!\cdot\ldots \cdot f_t!\cdot\prod_{j=0}^t H(j)^{(\mathbf{f}_\pi+\varepsilon\cdot \mathbf{m}^{(\pi)})_j}}{
(\mathbf{f}_\pi+\varepsilon\cdot \mathbf{m}^{(\pi)})_0!\cdot\ldots\cdot (\mathbf{f}_\pi+\varepsilon\cdot \mathbf{m}^{(\pi)})_t!\cdot\prod_{j=0}^t H(j)^{f_j}}\right\rbrace
\]
\item[\textbf{Step 5}]: Perform a standard Monte Carlo sampling over the orbit $\pi$ induced by $\mathbf{f}_{\pi}$ and compute the corresponding Monte Carlo cdf $\hat{F}_U(u_{obs}\vert \pi)$
\item[\textbf{Step 6}]: Repeat steps 2 to 5 $N_{sim}$ times.
\item[\textbf{Step 7}]: Compute an estimate of the distribution of $U$ under $H_0$ as the average $$\frac{1}{N_{sim}}\sum_{i=1}^{N_{sim}}\hat{F}_U(u_{obs}\vert \pi_i).$$
\end{description}

We point out that the Monte Carlo samplings in Steps 1 and 5 do not affect the convergence to the exact distribution because of the well-known properties of the Monte Carlo cumulative distribution function. In fact the Monte Carlo cdf is an unbiased estimator of $F_U(u\vert\pi)$, the cdf over the whole orbit $\pi$, which converges uniformly to $F_U(u\vert\pi)$ \citep{casella}.

In Section \ref{sec:compare} we compare the two estimators of $F_U(u\vert \fiber)$, the distribution of $U$ given $T$ under $H_0$, based on the indicator function $\mathbb{I}_{(U(\mathbf{y})\leq u)}(\mathbf{y})$ used in the algorithm described in Section \ref{algo:mcmc} and the distribution over one permutation orbit $\pi$, $F_U(u\vert \pi)$ used in the algorithm described in Section \ref{algo:mcmc_orbit}.

\section{Comparison of Estimators}
\label{sec:compare}
As stated above, we consider two estimators of $F_U(u\vert \fiber)$: the indicator function $\mathbb{I}_{(U(\mathbf{y})\leq u)}(\mathbf{y})$, which is used in the MCMC sampling presented in Section \ref{algo:mcmc}, and $F_U(u\vert \pi)$, used in the MCMC sampling presented in Section \ref{algo:mcmc_orbit}.
We are going to prove that the second estimator has lower dispersion in terms of variance and mean absolute deviation (MAD) than the other.

First we prove that both estimators are unbiased.
\begin{proposition}
\label{prop} $\mathbb{E}_{p}\left[\mathbb{I}_{(U(\mathbf{y})\leq u)}(\mathbf{y})\right]=\mathbb{E}_{p_\pi}\left[F_U(u\vert \pi)\right]=F_U(u\vert \fiber)$.
\end{proposition}
\begin{proof}
If we compute the expectation of $F_U(u\vert \pi)$ using $p_\pi$ we get
\begin{align*}
\mathbb{E}_{p_\pi}\left[F_U(u\vert \pi)\right]&=\sum_{\pi\subseteq\mathcal{F}_{N,t}} p_\pi(\pi) F_U(u\vert \pi) \\
&=\sum_{\pi\subseteq\mathcal{F}_{N,t}} p_\pi(\pi)\frac{1}{\# \pi}\sum_{\mathbf{y}\in \pi}\mathbb{I}_{(U(\mathbf{y})\leq u)}(\mathbf{y})\\
&=\sum_{\pi\subseteq\mathcal{F}_{N,t}} p_\pi(\pi)\sum_{\mathbf{y}\in \pi}\mathbb{I}_{(U(\mathbf{y})\leq u)}(\mathbf{y})\mathbb{P}(\mathbf{y}\vert\pi)\\
&=\sum_{\pi\subseteq\mathcal{F}_{N,t}}\sum_{\mathbf{y}\in \pi}\mathbb{I}_{(U(\mathbf{y})\leq u)}(\mathbf{y})p(\mathbf{y})\\
&=\sum_{\mathbf{y}\in\mathcal{F}_{N,t}} p(\mathbf{y})\mathbb{I}_{(U(\mathbf{y})\leq u)}(\mathbf{y})\\
&=\mathbb{E}_{p}\left[\mathbb{I}_{(U(\mathbf{y})\leq u)}(\mathbf{y})\right]\\
&=F_U(u\vert \fiber).
\end{align*}
\end{proof}
\smartqed \qed

As a consequence of Proposition \ref{prop}, we can estimate the distribution $F_U(u\vert \fiber)$ both with the weighted average of the values $\mathbb{I}_{(U(\mathbf{y})\leq u)}$ for a sufficiently large number of points $\mathbf{y}$, where the weights are the estimated probabilities $p(\mathbf{y})$, and with the weighted average of the values $F_U(u\vert \pi)$ for a sufficiently large number of orbits $\pi$, where the weights are the estimated probabilities $p_\pi(\pi)$, and that both estimators are unbiased.

The first approach, i.e. using  $\mathbb{I}_{(U(\mathbf{y})\leq u)}$ as estimator, is performed into the MCMC algorithm presented in Section \ref{algo:mcmc}, while the second one corresponds to the algorithm in Section \ref{algo:mcmc_orbit}.

We consider two measures of dispersion (variance and MAD) for the two estimators.
The following results hold.
\begin{proposition}
\label{prop2}
$\variance_{p}\left[\mathbb{I}_{(U(\mathbf{y})\leq u)}\right] \geq \variance_{p_\pi} \left[F_U(u\vert \pi)\right]$.
\end{proposition}

We observe that $\mathbb{I}_{(U(\mathbf{y})\leq u)}$ follows a Bernoulli distribution with parameter $F_U(u\vert \fiber)$. There follows that the variance of $\mathbb{I}_{(U(\mathbf{y})\leq u)}$ is $F_U(u\vert \fiber) \left(1-F_U(u\vert \fiber)\right)$. 

\begin{proposition}
\label{prop3}
$\mad_{p}\left[\mathbb{I}_{(U(\mathbf{y})\leq u)}\right]\geq\mad_{p_\pi}\left[F_U(u\vert\pi)\right]$.
\end{proposition}
The proof of Proposition \ref{prop2} and \ref{prop3} are presented in the Appendix \ref{app:estimator}.

Having proved that both dispersion measures are lower for the estimator $F_U(u\vert \pi)$, we conclude that the MCMC proposed in Section \ref{sec:permutation1} gives more reliable estimates of $F_U(u\vert\fiber)$.

\section{Permutation and MCMC sampling}
\label{sec:permutation2}
In this section we carry out a brief analysis of the limit case which we get when we sample just one orbit $\pi$ and we carry out a standard Monte Carlo sampling over $\pi$. If the sampled orbit is $\pi_{\mathbf{y}_{obs}}$, i.e. the one which contains the observed vector $\mathbf{y}_{obs}$, the sampling procedure proposed in Section \ref{sec:permutation1} corresponds to the standard permutation sampling (\cite{pesarin2010permutation}).

We observe that $\mathbf{y}_{obs}$ is an observation sampled from the distribution $p$ and that the corresponding orbit $\pi_{\mathbf{y}_{obs}}$ is an observation sampled from the distribution $p_\pi$, where $p$ and $p_\pi$ are the probability distributions in Proposition \ref{prop}.

Two well-known remarkable properties of the permutation sampling immediately follows from Proposition \ref{prop}. First $F_U(u\vert \pi_{\mathbf{y}_{obs}})$, the cumulative distribution function conditional to the orbit of the observed sample, is an unbiased estimator of $F_U(u\vert \fiber)$, the cumulative distribution function over the fiber $\mathcal{F}_{N,t}$. Secondly, it is an unbiased estimator of $F_U(u\vert \fiber)$ for \emph{any} distribution function $p$, that does not need to be specified. In fact the estimator $F_U(u\vert \pi_{\mathbf{y}_{obs}})$
does not need any expression of $p$ to be computed.

As a simple example consider again the fiber $\mathcal{F}_{3,6}$ in Figure \ref{fig:grafo}. We select $n_1=2$ and $n_2=1$ and we compare the exact cumulative distribution over the fiber $F_U(u\vert \fiber)$ and the cumulative distribution over $\pi=\pi_{(1,2,3)}$, the orbit with highest probability, $F_U(u\vert \pi)$. We get two distributions (Table \ref{tab:finale}) which are considerably close, even if the cardinality of the selected orbit is low ($\#\pi_{(1,2,3)}=6$) compared to the the cardinality of $\mathcal{F}_{3,6}$, which is $28$. However, it is easy to see that some orbits do not give a good approximation of the distribution over $\fiber$. If we refer again to the fiber $\mathcal{F}_{3,6}$ and we consider $\pi_{(2,2,2)}$, we get a cumulative distribution which has only two values, 0 and 1 (Table \ref{tab:finale}). This difference is due to the unequal probabilities of the orbits in $\fiber$ (these probabilities are reported in Table \ref{tab:7orbits}).

\begin{table}[ht]
\centering
\caption{Cumulative distributions of $U$}
\label{tab:finale}    
\begin{tabular}{llllllll}
\hline\noalign{\smallskip}
$u$ & 0&1	&2		&3	&4	&5	&6\\
\noalign{\smallskip}\hline\noalign{\smallskip}
$\mathcal{F}_{3,6}$ &0.001	&0.018	&0.100	&0.320	&0.649	&0.912	&1\\
$\pi_{(1,2,3)}$& 0&0 &0 &0.333&0.667&1&1\\
$\pi_{(2,2,2)}$& 0&0 &0 &0&1&1&1\\
\noalign{\smallskip}\hline
\end{tabular}
\end{table}


\section{Simulation Study}
\label{sec:sim}

We compare the \emph{approximated} conditional cumulative distribution functions obtained using the sampling techniques described in Sections \ref{sec:mcmc}, \ref{sec:permutation1} and \ref{sec:permutation2} with the \emph{exact} conditional cumulative distribution function in the case of Poisson data. A preliminary simulation study is presented in \cite{sis}.

We consider Poisson distributed data: $\mathbf{Y}_1^{(n_1)}=(Y_{1},\ldots, Y_{n_1})$ of size $n_1$ from $\Poi(\mu_1)$ and $\mathbf{Y}_2^{(n_2)}=(Y_{n_1+1},\ldots, Y_{n_1+n_2})$ of size $n_2$ from $\Poi(\mu_2)$.
In this case the exact distribution \eqref{eq:umpu} under $H_0:\mu_1=\mu_2=\mu$ is known to be a binomial distribution with $t$ trials and probability of success $n_1/(n_1+n_2)$.

We consider $9$ scenarios built taking three different sample sizes $(n_1,n_2)$ (Table \ref{tab:samplesize}) and, for each sample size, three different population means $(\mu_1,\mu_2)$ (Table \ref{tab:lambda}).  

\begin{table}[th]

\centering
\caption{Scenario definition}
\label{tab:scenario}
\subfloat[][Sample sizes \label{tab:samplesize}]{
\begin{tabular}{llll}
\hline\noalign{\smallskip}
 & $1$ & $2$ & $3$ \\
\noalign{\smallskip}\hline\noalign{\smallskip}
$n_1$ & $6$ & $10$ & $30$ \\
$n_2$ & $4$ & $15$ & $20$ \\
\noalign{\smallskip}\hline
\end{tabular}
}
\qquad\qquad
\subfloat[][Population means \label{tab:lambda}]{
\begin{tabular}{llll}
\hline\noalign{\smallskip}
 & $1$ & $2$ & $3$ \\
\noalign{\smallskip}\hline\noalign{\smallskip}
$\mu_1$ & $1$ & $1$ & $1$\\
$\mu_2$ & $1$ & $1.5$ & $2$\\
\noalign{\smallskip}\hline
\end{tabular}
}
\end{table}

First we compare how fast the two MCMC procedures converge to the true distribution $F_U(u\vert \fiber)$.
We draw one random sample $\mathbf{y}_{obs}$ for each scenario above, we run both MCMC procedures (the one which samples vectors $\mathbf{y}\in\fiber$ and the one which samples orbits $\pi\in\mathcal{F}_{N,t}^{(\pi)}$) for 5,000 steps and at each step we compute the corresponding estimate of $F_U(u\vert\fiber)$ (the indicator function $\mathbb{I}_{(U(\mathbf{y})\leq u)}$ in the first case and the permutation distribution $F_U(u\vert\pi)$ in the second one).

The number of Monte Carlo permutations for every sampled orbit $\pi$ is given by
\begin{equation}
\label{eq:rule}
r(\pi)=w_{\pi}\cdot 100,000,
\end{equation}
where $w_{\pi}=\#\pi/\#\fiber$ is the weight of the sampled orbit with respect to the total cardinality of the fiber. The two cardinalities are computed as $\#\fiber=\binom{t+N-1}{N-1}$ and $\#\pi$ as in equation \eqref{eq:exact_card_orbit}. This choice allows us to get a suitable total number of permutations and results in a number of Monte Carlo samples proportional to the cardinality of the orbit.

\begin{figure*}[t]
\centering
\subfloat[$n_1=6$, $n_2=4$, $\mu_1=1$, $\mu_2=1$]{
\includegraphics[width=0.5\textwidth]{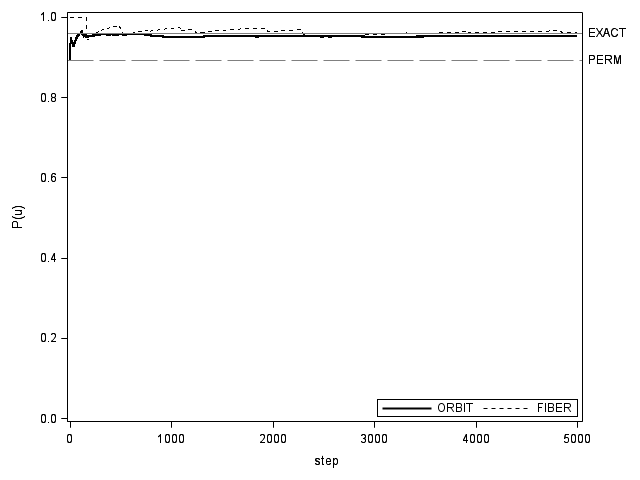}}
\subfloat[$n_1=6$, $n_2=4$, $\mu_1=1$, $\mu_2=1.5$]{
\includegraphics[width=0.5\textwidth]{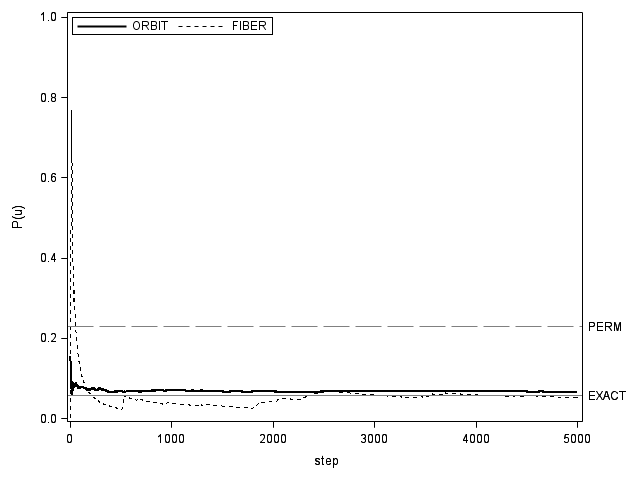}}\\
\subfloat[$n_1=10$, $n_2=15$, $\mu_1=1$, $\mu_2=1$]{
\includegraphics[width=0.5\textwidth]{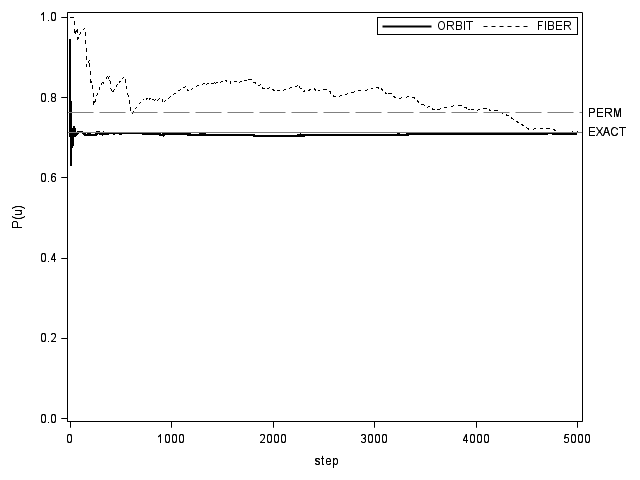}}
\subfloat[$n_1=10$, $n_2=15$, $\mu_1=1$, $\mu_2=1.5$]{
\includegraphics[width=0.5\textwidth]{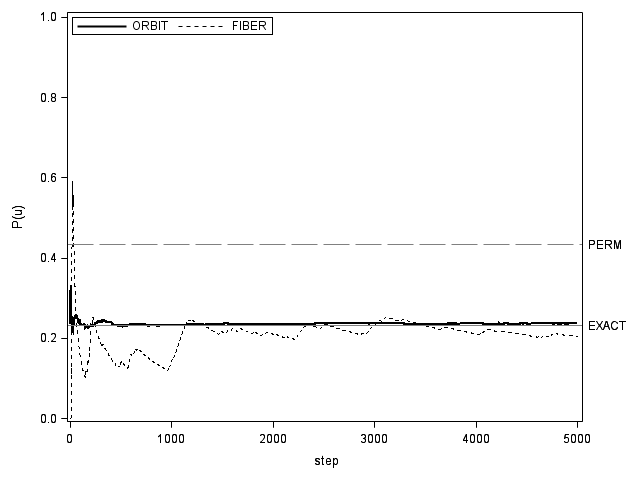}}
\caption{Convergence comparison for 5,000 steps}
\label{fig:cfr}
\end{figure*}

Figure \ref{fig:cfr} shows four examples of the behaviour of the two MCMC procedures.
Much faster convergence to the value computed using the exact cdf (solid horizontal line) is obtained with the MCMC over the orbits of permutations (solid line) compared to the estimated cdfs obtained through the accelerated MCMC sampling over the elements of the fiber (dashed line).
In the same figure, we report $\hat{F}_U(u\vert\pi_{\mathbf{y}_{obs}})$, the Monte Carlo permutation estimate of $F_U(u\vert \fiber)$ (dashed horizontal line), too. The number of Monte Carlo permutations for $\pi_{\mathbf{y}_{obs}}$ is $10,000$.

Now we compare the execution times for 5,000 steps for the three sampling algorithms. We consider $1,000$ samples with $n_1=30$ and $n_2=20$ for each population mean in Table \ref{tab:lambda} and we run both MCMC samplings for 5,000 steps. We report the average execution times (in seconds) for the three sampling procedures (MCMC over the fiber, MCMC over the orbits and standard permutation) in Table \ref{tab:times}.

\begin{table}[h]
\centering
\caption{Average execution times (in seconds) for $n_1=30$ and $n_2=20$ for 5,000 steps}
\begin{tabular}{ccccc}
\hline\noalign{\smallskip}
\multicolumn{2}{c}{Scenario} & MCMC fiber & MCMC orbit & Permutation \\
$\mu_1$&$\mu_2$&&&\\
\noalign{\smallskip}\hline\noalign{\smallskip}
1&1&0.32&10.79&$2.2\cdot10^{-4}$\\
1&1.5&0.30&17.58&$7.0\cdot10^{-5}$\\
1&2&0.32&28.01&$5.9\cdot10^{-5}$\\
\noalign{\smallskip}\hline
\end{tabular}
\label{tab:times}
\end{table}

As expected, the MCMC sampling over the set of orbits $\mathcal{F}_{N,t}^{(\pi)}$ takes more time than the MCMC sampling over $\mathbf{y}\in\fiber$ to perform 5,000 steps. This is due to the fact that every time an orbit $\pi$ is sampled from $\mathcal{F}_{N,t}^{(\pi)}$ we perform a standard Monte Carlo over $\pi$, hence a Monte Carlo sampling corresponds to every step of the sampling over $\mathcal{F}_{N,t}^{(\pi)}$. In any case the maximum observed execution time (28.01 seconds) is very low. Moreover the achieved convergence speed seems to justify the extra-time needed by the MCMC over the orbits. 

Indeed, if we consider the convergence speed (Figure \ref{fig:cfr}), we observe that the MCMC sampling over the orbits, being computationally simple, is very fast and needs less than 1,000 steps to converge to the exact value. On the contrary the MCMC sampling over $\mathbf{y}\in\fiber$ is much less efficient, in fact its convergence to the exact value is not satisfactory even in 5,000 steps. We observe that at least 10,000 steps are needed to have a good convergence for the MCMC over $\mathbf{y}\in\fiber$ (see Figure \ref{fig:cfrbis}), even if this number of steps does not always give a completely satisfactory convergence, as we can see in Figure \ref{fig:3c}.

\begin{figure*}[t]
\centering
\subfloat[$n_1=6$, $n_2=4$, $\mu_1=1$, $\mu_2=1$]{
\includegraphics[width=0.5\textwidth]{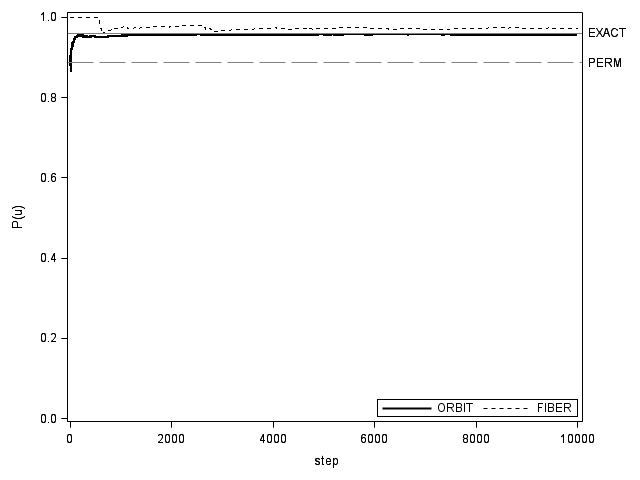}}
\subfloat[$n_1=6$, $n_2=4$, $\mu_1=1$, $\mu_2=1.5$]{
\includegraphics[width=0.5\textwidth]{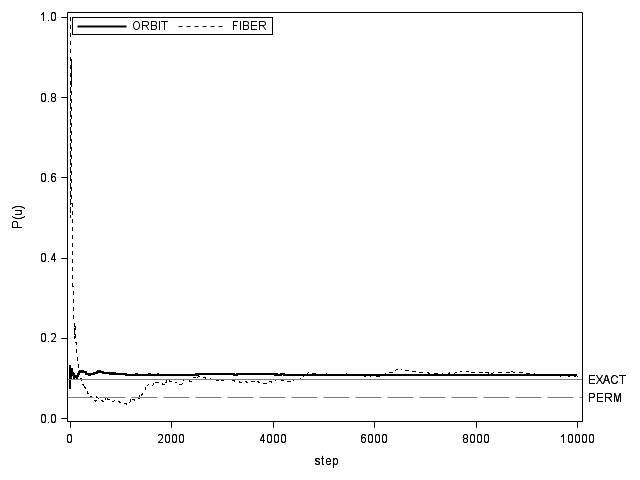}}\\
\subfloat[$n_1=10$, $n_2=15$, $\mu_1=1$, $\mu_2=1$ \label{fig:3c}]{
\includegraphics[width=0.5\textwidth ]{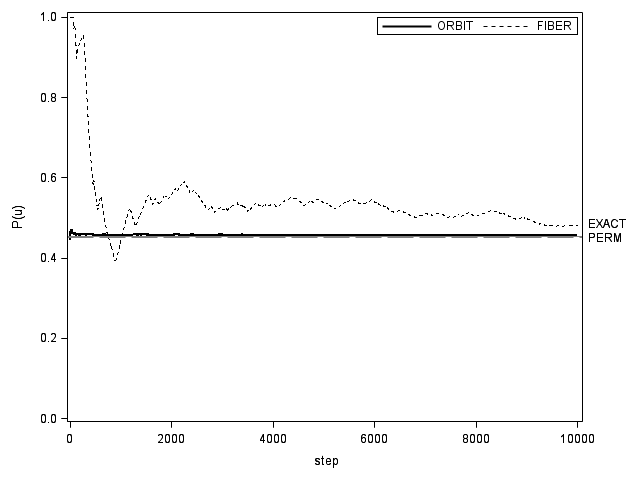}}
\subfloat[$n_1=10$, $n_2=15$, $\mu_1=1$, $\mu_2=1.5$]{
\includegraphics[width=0.5\textwidth]{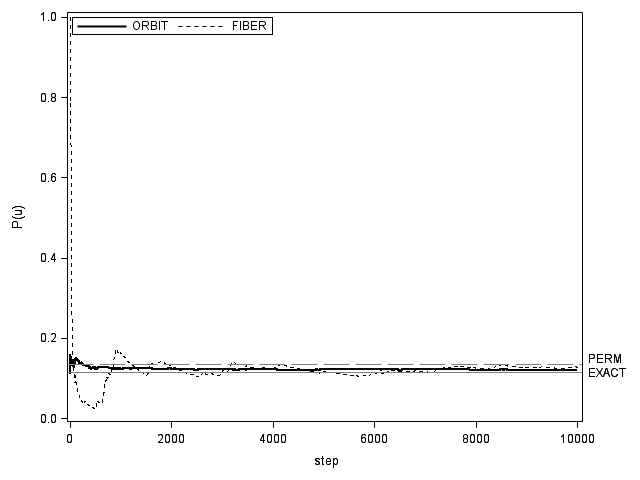}}
\caption{Convergence comparison for 10,000 steps}
\label{fig:cfrbis}
\end{figure*}

To further explore the convergence of the MCMC sampling over the orbits we consider the errors of the approximated distributions obtained with the three sampling algorithms (MCMC over $\fiber$, MCMC over $\mathcal{F}_{N,t}^{(\pi)}$, Monte Carlo permutation in $\pi_{\mathbf{y}_{obs}}$) with respect to the exact value.

Given that we work with Poisson distributed data we know that the exact value is given by
\begin{equation}
\label{eq:exact_poisson}
\mathbb{P}(\Bin(t,\theta_0)\leq u)=\sum_{k=0}^u \binom{t}{k}\theta_0^k(1-\theta_0)^{t-k},
\end{equation}
with $\theta_0=n_1/(n_1+n_2)$.

We consider again the 9 scenarios above and for each scenario we draw 1,000 samples. We perform 1,000 burn-in steps plus 4,000 actual steps for the MCMC sampling over the fiber, while for the MCMC over the orbits we use 250 burn-in steps followed by 850 actual steps. The reduction of the number of steps for the MCMC over the orbits is justified by the fact that this MCMC achieves satisfactory convergence in less steps (see Figure \ref{fig:cfr}).

For each sample we compute the errors of the three estimated distributions (we denote by $\hat{F}_U^{\pi}$ the estimated distribution obtained by sampling orbits and $\hat{F}_U^{\mathbf{y}}$ the estimated distribution obtained by sampling vectors $\mathbf{y}$)
\begin{align*}
&\bias_{\mathbf{y}}=\sum_{k=0}^u \binom{t}{k}\theta_0^k(1-\theta_0)^{t-k}-\hat{F}_U^{\mathbf{y}}(u\vert\fiber),\\
&\bias_{\pi}=\sum_{k=0}^u \binom{t}{k}\theta_0^k(1-\theta_0)^{t-k}-\hat{F}_U^{\pi}(u\vert\mathcal{F}_{N,t}^{(\pi)}),\\
&\bias_{\text{perm}}=\sum_{k=0}^u \binom{t}{k}\theta_0^k(1-\theta_0)^{t-k}-\hat{F}_U(u\vert\pi_{\mathbf{y}_{obs}}).
\end{align*}
We know from Proposition \ref{prop} that the expected value of $\bias_{\mathbf{y}}$, $\bias_{\pi}$ and $\bias_{\text{perm}}$ is 0. 

\begin{table*}[t]
\centering
\caption{Error analysis}
\subfloat[]{
\begin{tabular}{cccc|cccccc}
\hline\noalign{\smallskip}
\multicolumn{4}{c}{Scenario}& \multicolumn{3}{c}{Mean}&\multicolumn{3}{c}{Range}\\
$n_1$&$n_2$&$\mu_1$&$\mu_2$&$\fiber$ & $\mathcal{F}_{N,t}^{(\pi)}$ & $\pi_{\mathbf{y}_{obs}}$&$\fiber$ & $\mathcal{F}_{N,t}^{(\pi)}$ & $\pi_{\mathbf{y}_{obs}}$\\
\noalign{\smallskip}\hline\noalign{\smallskip}
6 &	4 &	1 &	1 & 0 & -0.004 & 0.007 & 0.217 & 0.062 & 0.54 \\ 
6 &	4 &	1 &	1.5& -0.002 & 0.002 & 0.011 & 0.205 & 0.057 & 0.436 \\ 
6 &	4 &	1 &	2& -0.001 & 0.005 & 0.015 & 0.212 & 0.05 & 0.442 \\ 
10 &	15 &	1 &1 &0.005 & -0.001 & -0.003 & 0.339 & 0.052 & 1.209 \\ 
10 &	15 &	1&1.5& 0 & 0.004 & 0.003 & 0.342 & 0.052 & 1.756 \\ 
10 &	15 &	1&2& 0 & 0.003 & 0.008 & 0.325 & 0.064 & 1.685 \\ 
30 &	20 &	1&1& 0 & -0.001 & 0.009 & 0.586 & 0.066 & 1.9 \\ 
30 &	20 &	1&1.5& 0.002 & 0.003 & 0.013 & 0.502 & 0.086 & 1.762 \\ 
30 &	20 &	1&2& 0 & 0.001 & -0.004 & 0.304 & 0.065 & 1.748 \\ 

\noalign{\smallskip}\hline
\end{tabular}}\\
\subfloat[]{
\begin{tabular}{cccc|cccccc}
\hline\noalign{\smallskip}
\multicolumn{4}{c}{Scenario}& \multicolumn{3}{c}{Std Dev}& \multicolumn{3}{c}{MAD}\\
$n_1$&$n_2$&$\mu_1$&$\mu_2$&$\fiber$ & $\mathcal{F}_{N,t}^{(\pi)}$ & $\pi_{\mathbf{y}_{obs}}$&$\fiber$ & $\mathcal{F}_{N,t}^{(\pi)}$ & $\pi_{\mathbf{y}_{obs}}$\\
\noalign{\smallskip}\hline\noalign{\smallskip}
6 &	4 &	1 &	1 & 0.031 & 0.01 & 0.051 & 0.023 & 0.009 & 0.036 \\ 
6 &	4 &	1 &	1.5& 0.028 & 0.01 & 0.051 & 0.021 & 0.008 & 0.037 \\ 
6 &	4 &	1 &	2& 0.026 & 0.008 & 0.05 & 0.018 & 0.007 & 0.038 \\ 
10 &	15 &1&	1& 0.049 & 0.008 & 0.072 & 0.037 & 0.006 & 0.038 \\ 
10 &	15 &1&	1.5& 0.043 & 0.007 & 0.129 & 0.031 & 0.006 & 0.065 \\ 
10 &	15 &1&	2& 0.03 & 0.007 & 0.148 & 0.018 & 0.005 & 0.067 \\ 
30 &	20 &	1&1& 0.074 & 0.009 & 0.208 & 0.056 & 0.007 & 0.121 \\ 
30 &	20 &	1&1.5& 0.053 & 0.009 & 0.234 & 0.035 & 0.007 & 0.13 \\ 
30 &	20 &	1&2& 0.024 & 0.006 & 0.121 & 0.011 & 0.003 & 0.037 \\ 

\noalign{\smallskip}\hline
\end{tabular}
}
\label{tab:bias}
\end{table*}

In Table \ref{tab:bias} we report some statistics about the errors observed for 1,000 random samples drawn for each scenario. We observe that on average the errors are close to 0 for every sampling procedure, as expected.
Then we consider the variability of the obtained errors. The range of variation of the three errors is considerably different: the MCMC sampling on the orbits of permutations gives a range which is one order of magnitude less than the others.  
We also report the standard deviation and the mean absolute deviation (MAD). Also for these statistics it is easy to see that the distribution $\hat{F}_U^{\pi}(u\vert\mathcal{F}_{N,t}^{(\pi)})$, obtained by MCMC sampling from $\mathcal{F}_{N,t}^{(\pi)}$, is the one with smaller variation for every scenario with both standard deviation and MAD of order $10^{-3}$.

The standard Monte Carlo permutation distribution is the one which gives the highest values for all the variability measures (at least one order of magnitude bigger than those obtained by $\hat{F}_U^{\pi}(u\vert\mathcal{F}_{N,t}^{(\pi)})$). This result is not unexpected, given that the standard Monte Carlo permutation over $\pi_{\mathbf{y}_{obs}}$ corresponds to the first step of the MCMC sampling over $\fiber^{(\pi)}$.

Both standard deviation and MAD for the distribution $\hat{F}_U^{\mathbf{y}}(u\vert\fiber)$ are of order $10^{-2}$, hence the MCMC sampling over $\fiber$ gives approximations whose goodness is half-way between the others.

These results are consistent with Propositions \ref{prop2} and \ref{prop3} in Section \ref{sec:compare}, and are confirmed by Figure \ref{fig:bias}, in which the histograms of the absolute errors for some of the scenario in Table \ref{tab:scenario} are shown.
\begin{figure*}[t]
\centering
\subfloat[$n_1=6$, $n_2=4$, $\mu_1=1$, $\mu_2=1$]{
\includegraphics[width=0.5\textwidth]{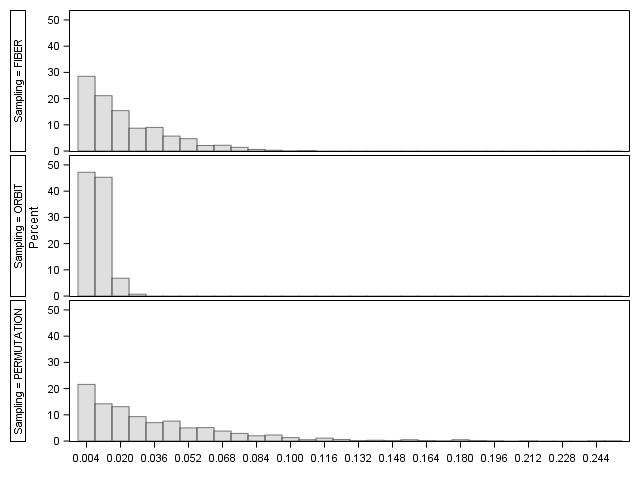}}
\subfloat[$n_1=6$, $n_2=4$, $\mu_1=1$, $\mu_2=1.5$]{
\includegraphics[width=0.5\textwidth]{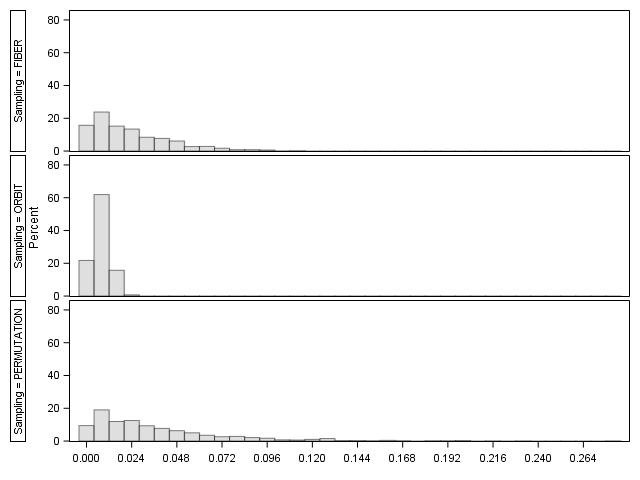}}\\
\subfloat[$n_1=10$, $n_2=15$, $\mu_1=1$, $\mu_2=1$]{
\includegraphics[width=0.5\textwidth]{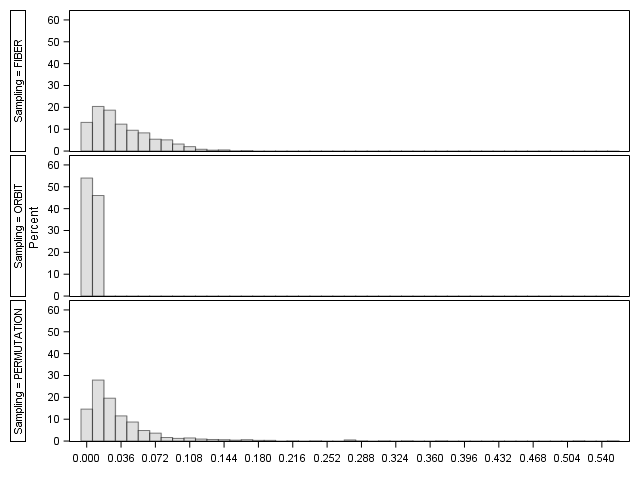}}
\subfloat[$n_1=10$, $n_2=15$, $\mu_1=1$, $\mu_2=1.5$]{
\includegraphics[width=0.5\textwidth]{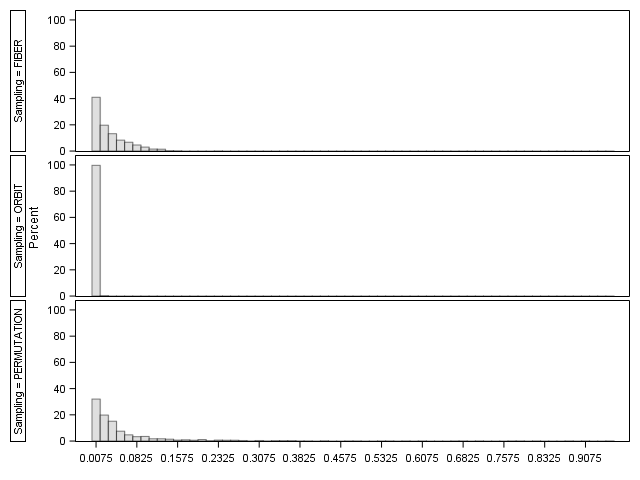}}
\caption{Observed absolute values of $\bias_{\mathbf{y}}$, $\bias_{\pi}$ and $\bias_{\text{perm}}$}
\label{fig:bias}
\end{figure*}

If we compare the execution times for the three sampling procedures (Table \ref{tab:times_full}) we observe that, as expected, despite the smaller number of steps (1,100 in total), the MCMC over $\fiber^{(\pi)}$ has the highest execution time, due to the above mentioned Monte Carlo sampling over the orbits.
Nevertheless the average time needed to perform one run of the MCMC over the orbits is always less than 7 seconds on a standard laptop, which makes this sampling procedure manageable in most real-world situations. 

Taking a closer look to the results obtained by the two MCMC procedures (the statistics on errors in Table \ref{tab:bias} and their execution times in Table \ref{tab:times_full}), we observe that both procedures have competitive execution times, but only in the MCMC over the orbits we can make use of the estimator $F_U(u\vert\pi)$, which has better properties in terms of dispersion, as stated in Proposition \ref{prop2} and \ref{prop3}, because this estimator is built exploiting the partition of the sample space $\fiber$ into orbits of permutations.

\begin{table}[h]
\centering
\caption{Average execution times (in seconds)}
\begin{tabular}{cccc|ccc}
\hline\noalign{\smallskip}
\multicolumn{4}{c}{Scenario} & $\fiber$ & $\mathcal{F}_{N,t}^{(\pi)}$ & $\pi_{\mathbf{y}_{obs}}$ \\
$n_1$&$n_2$&$\mu_1$&$\mu_2$&&&\\
\noalign{\smallskip}\hline\noalign{\smallskip}
6 &	4 &	1 &	1 & 0.105 &4.173&$5.4\cdot 10^{-3}$\\
6 &	4 &	1 &	1.5& 0.111&2.441&$3.5\cdot 10^{-3}$\\
6 &	4 &	1 &	2& 0.118&1.479&$2.3\cdot 10^{-3}$\\ 
10 &	15 &1&	1&0.113&1.262&$1.3\cdot 10^{-3}$\\ 
10 &	15 &1&	1.5& 0.123&1.062&$3.8\cdot 10^{-4}$\\ 
10 &	15 &1&	2&0.133&1.529&$1.1\cdot 10^{-4}$\\ 
30 &	20 &	1&1&0.122&2.751&$1.9\cdot 10^{-4}$\\ 
30 &	20 &	1&1.5&0.123&4.282&$6.2\cdot 10^{-5}$\\ 
30 &	20 &	1&2& 0.127&6.651&$3.2\cdot 10^{-5}$\\ 
\noalign{\smallskip}\hline
\end{tabular}
\label{tab:times_full}
\end{table}

Thus the MCMC over the orbits, despite the highest execution time, allows us to use an estimator which is theoretically more reliable. The consequences of this fact are visible in Table \ref{tab:bias}, indeed range of variation, standard deviation and MAD for $E_\pi$ are the lowest, as we observed earlier.

Additionally, with small changes in the proposed code, it is possible to reduce the computational times listed in Table \ref{tab:times_full} for the MCMC over $\fiber^{(\pi)}$:
\begin{enumerate}
\item instead of using the rule $r(\pi)$ in equation \eqref{eq:rule}, which requires the computation of the cardinality $\#\pi$ for every sampled orbit $\pi$, one could fix a priori the number of Monte Carlo replications;
\item the number of burn-in steps can be reduced, as suggested by Figure \ref{fig:cfr} and \ref{fig:cfrbis}.
\end{enumerate}

The second point is justified by observing that the MCMC sampling which exploits the partition into orbits of permutations can be seen as a \emph{super-accelerated} version of the standard algorithm proposed by \cite{sturmfels}. In fact, both the standard version and the accelerated one described in Section \ref{sec:mcmc} need some steps to drift from $\mathbf{y}_{obs}$ and reach stationarity (the so-called burn-in steps) because at every step one move only is selected, hence every step corresponds to moving in one direction, either by one jump, for the standard version, or by more jumps, for the accelerated version.

On the contrary, performing the MCMC over the orbits we jump from one permutation orbit to another, allowing for much wider jumps which make the underlying Markov chain losing the dependence from the starting state $\mathbf{y}_{obs}$ earlier and result in a faster convergence to the stationary distribution.
Moreover the number of vectors sampled at each step of the three MCMC algorithms is significantly different:
for the standard and accelerate versions only one vector $\mathbf{y}$ is sampled at each step, while for the MCMC over the orbits the number of vectors sampled at each step depends on the Monte Carlo replications $r(\pi)$ sampled from each orbit $\pi$.

These two characteristics of the MCMC over the orbits clearly speed up the convergence to the stationary distribution.

As a simple example let us consider the fiber $\mathcal{F}_{3,6}$ in Figure \ref{fig:grafo}. Suppose $\mathbf{y}_{obs}=(6,0,0)$, the vertex on the top left. If we perform one step of the standard MCMC algorithm over $\mathcal{F}_{3,6}$ we can move towards $(5,1,0)$ or $(5,0,1)$, depending on the selected move ($\mathbf{m}_1$ and $\mathbf{m}_2$ in equation \eqref{eq:b3} respectively). Hence we have 2 possible jumps, one for each move, and we can reach 2 elements of $\mathcal{F}_{3,6}$.

If we perform one step of the accelerated MCMC algorithm over $\mathcal{F}_{3,6}$ we can move towards any of the vectors $(5,1,0) \ldots (0,6,0)$ on the right or any of the vectors $(5,0,1) \ldots (0,0,6)$ below, depending on the selected move ($\mathbf{m}_1$ and $\mathbf{m}_2$ in equation \eqref{eq:b3} respectively), but not to any of the 15 vectors on the bottom right of $(6,0,0)$ (e.g $(4,1,1)$ or $(3,1,2)$). Hence we have 12 possible jumps, 6 for each move, and we can reach 12 elements of $\mathcal{F}_{3,6}$.

Let us now consider the Figure \ref{fig:grafo_orbite} where the graph whose nodes are the orbits in $\mathcal{F}_{3,6}^{(\pi)}$ is shown. It is evident that this graph is simpler than the one shown in Figure \ref{fig:grafo}, i.e. the graph over $\mathcal{F}_{3,6}$. There are 7 nodes instead of 28 and 9 edges instead of 42. If we perform one step of the MCMC over the set of orbits $\mathcal{F}_{3,6}^{(\pi)}$ contained in $\mathcal{F}_{3,6}$, we can jump from $\pi_{(6,0,0)}$ to $\pi_{(0,1,5)}$ or $\pi_{(0,2,4)}$ or $\pi_{(0,3,3)}$, depending on the selected move ($\mathbf{m}_7^{(\pi)}$, $\mathbf{m}_8^{(\pi)}$ and $\mathbf{m}_9^{(\pi)}$ in equation \eqref{eq:b6} respectively). Hence we have 3 possible jumps, one for each move. However if we look at the vectors contained in each orbit we observe that starting from $\pi_{(6,0,0)}$ we can reach 15 elements of $\mathcal{F}_{3,6}$.

\begin{figure*}
\begin{center}

\resizebox{0.5\textwidth}{!}{%
\begin{tikzpicture}[-latex,auto,node distance=2cm and 2cm,on grid, 
	semithick, state/.style={circle,top color=white,bottom color=white, 
  	draw, black, text=black,minimum width=1cm}] 
%
%
  	
\node[state] (AA) {$\pi_{033}$}; 
\node[state] (AB) [right=of AA] {$\pi_{123}$};
\node[state] (AC) [right=of AB] {$\pi_{024}$};
\node[state] (AD) [right=of AC] {$\pi_{222}$};

\node[state] (BB) [below=of AB] {$\pi_{006}$};
\node[state] (BC) [right=of BB] {$\pi_{015}$};
\node[state] (BD) [right=of BC] {$\pi_{114}$};

\path[-] (AA) edge node[above]{$\mathbf{m}^{(\pi)}_2$} (AB);
\path[-] (AB) edge node[above]{$\mathbf{m}^{(\pi)}_3$} (AC);
\path[-] (AC) edge node[above]{$\mathbf{m}^{(\pi)}_4$} (AD);
\path[-] (AA) edge node[above=3 pt]{$\mathbf{m}^{(\pi)}_9$} (BB);
\path[-] (AB) edge node[above=3 pt]{$\mathbf{m}^{(\pi)}_6$} (BC);
\path[-] (AC) edge node[above=3 pt]{$\mathbf{m}^{(\pi)}_1$} (BD);
\path[-] (BB) edge node[below]{$\mathbf{m}^{(\pi)}_7$} (BC);
\path[-] (BC) edge node[below]{$\mathbf{m}^{(\pi)}_5$} (BD);
\path[-] (BB) edge node[below=3 pt]{$\mathbf{m}^{(\pi)}_8$} (AC);

\end{tikzpicture}
}
\end{center} 
\caption{Graph on the orbits of $\mathcal{F}_{N,t}^{(\pi)}$ with $t=6$ and $N=3$}
\label{fig:grafo_orbite}
\end{figure*}
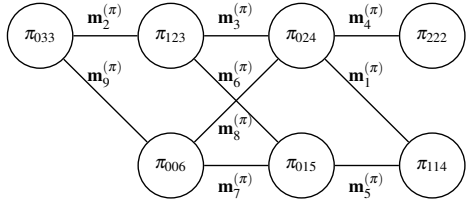

In conclusion, the simulation study provides strong evidence that the MCMC sampling over the set of orbits $\mathcal{F}_{N,t}^{(\pi)}$ (Section \ref{sec:permutation1}) gives estimates of the exact distribution in equation \eqref{eq:exact_poisson} which are considerably more reliable than those obtained by sampling vectors from $\fiber$ (Section \ref{sec:mcmc}), and at the same time the execution time required is still manageable. 

\paragraph{Computational Details}
The simulation study presented in this section was implemented in SAS/IML\textregistered.
The software code is available upon request.
We performed the analysis using a standard laptop (CPU Intel core 2 Duo T6570 CPU 2.10GHz
2.10GHz, RAM 4GB).
\section{Conclusions and Further Perspectives}
\label{sec:conclusion}

In this work we presented two Markov Chain Monte Carlo procedures to sample from the distribution in equation \eqref{eq:umpu} and perform approximated conditional tests.
These sampling procedures can be used to perform uniformly most powerful unbiased tests on the observed data, and are particularly valuable when a computational-friendly form of distribution (\ref{eq:umpu}) cannot be found or when the enumeration of the elements of the conditional sample space is infeasible, since the convergence to (\ref{eq:umpu}) is granted by the Metropolis-Hastings theorem \citep{casella}.

Our main contribution is a 2-step sampling algorithm (Section \ref{sec:compare}) which exploits the properties of the orbits of permutations to better approximate the exact distribution in equation \eqref{eq:umpu} under $H_0$. Specifically, this algorithm grants a faster convergence to the exact distribution if compared to the standard MCMC algorithm proposed by \cite{sturmfels}. At the same time it gives more reliable estimates in a reasonable computational time.

Furthermore the MCMC sampling procedure based on orbits of permutations establishes a link between standard permutation and algebraic-statistics-based sampling that, to the best of our knowledge, is unknown. 

In addition the sampling strategy presented in Section \ref{sec:permutation1} can be used to perform other types of hypothesis tests. We consider possible extensions to N-way ANOVA models, see also \cite{aokitakemura}. First we focus on the 1-way ANOVA model with $k\geq2$ levels each one with $n_i \geq 1, i=1,\ldots,k$ replications. As an example we consider $k=3$. The transpose of the full-rank design matrix $\mathbf{X}$ is
\[
\mathbf{X}^T=\begin{bmatrix*}[r]
1&\ldots&1&1&\ldots&1&1&\ldots&1\\
1&\ldots&1&0&\ldots&0&-1&\ldots&-1\\
0&\ldots&0&1&\ldots&1&-1&\ldots&-1\\
\end{bmatrix*}.\]
If we are interested in testing $$H_0:\beta_1=\ldots=\beta_{k-1}=0\quad \text{vs}\quad H_1:\text{at least one $\beta$ is not 0},$$ the nuisance parameter of the test is $\beta_0$. Given that the sufficient statistic for $\beta_0$ under $H_0$ is $T=\sum_{i=1}^{n_1+\ldots+n_k} Y_i$, to perform an exact test we sample from the fiber $\fiber$ (the same defined in equation \eqref{eq:fiber}) with the MCMC presented in Section \ref{sec:permutation1}. In this case the test statistic is the usual ANOVA F-statistic $\text{MS}_{\text{groups}}/\text{MS}_{\text{residuals}}$.

Next we focus on N-way ANOVA models with N$\geq2$. If we consider $\psi(\mu_i)=\beta_0$ as the null model, the set we sample from under $H_0$ is again the fiber $\fiber$.
On the contrary if we consider null models with at least one effect or interaction, the set we need to sample from does not coincide with $\fiber$.
We show the two-way ANOVA model for 2-level factors ($k_j=2$, $j=1,2$) with no interaction $\psi(\mu_i)=\beta_0+\beta_1x_{i1}+\beta_2x_{i2}$ as an example. We consider the transpose of the full-rank design matrix $\mathbf{X}$
\[
\mathbf{X}^T=\begin{bmatrix*}[r]
1&\ldots&1&1&\ldots&1&1&\ldots&1&1&\ldots&1\\
1&\ldots&1&1&\ldots&1&-1&\ldots&-1&-1&\ldots&-1\\
1&\ldots&1&-1&\ldots&-1&1&\ldots&1&-1&\ldots&-1\\
\end{bmatrix*},\]
with $x_{i2}=1$ if $1\leq i\leq k_1< n_1$ or $n_1+1\leq i\leq k_2< n_1+n_2$, $x_{i2}=-1$ if $k_1+1\leq i\leq n_1$ or $k_2+1\leq i\leq n_1+n_2$.
The hypotheses are $H_0: \beta_2=0$ against $H_1: \beta_2\neq0$, hence $\beta_0$ and $\beta_1$ are the nuisance parameters.
The sufficient statistic for the nuisance parameters $\begin{pmatrix}\beta_0\\\beta_1\end{pmatrix}$ is $\mathbf{T}=\begin{pmatrix}
\sum_{i=1}^N Y_i\\
\sum_{i=1}^{n_1}Y_i-\sum_{i=1}^{n_1+n_2}Y_i
\end{pmatrix}$, and the corresponding conditional sample space is
\begin{equation}
\label{eq:second_ext}
\left\lbrace (Y_1,\dots,Y_{n_1+n_2})\in \mathbb{N}^{n_1+n_2}: \begin{pmatrix*}[r]
1&\ldots&1&1&\ldots&1\\
1&\ldots&1&-1&\ldots&-1\\
\end{pmatrix*}\mathbf{Y}= \mathbf{t}\right\rbrace.
\end{equation}

The set \eqref{eq:second_ext} can be partitioned into subsets containing only permutations and thus, thanks to the inferentially equivalence in equation \eqref{eq:inf_equi}, we can built a sampling procedure which uses both MCMC and Monte Carlo sampling as in Section \ref{sec:permutation1}. Nevertheless the subsets containing permutations will not coincide with the orbits $\pi$. This follows from the second condition in \eqref{eq:second_ext}, in fact the sample space \eqref{eq:second_ext} does not contain any full orbit $\pi$, but instead the subset of within-groups permutations $\omega$ for every orbit $\pi$.

The more efficient permutation-based sampling strategy can be extended to more complicated designs, provided that the sample size $N$ allows for a partition $\Omega=\lbrace\omega\rbrace$ where at least some $\omega$s contain a sufficiently large number of points. An example of an experiment in which the sample size $N$ does not allow for the permutation-based sampling strategy is presented in \cite{aokitakemura}.

\appendix
\section{Properties of the Graph}
\label{app:graph}
In this appendix we focus on the description of the graph $\mathcal{G}$ induced by the Markov basis over the fiber $\mathcal{F}_{N,t}$.

First we show how to obtain the formula in equation (\ref{eq:edges}) to compute the number of edges $\vert E\vert$ of $\mathcal{G}$. Then we will prove that the graph is bipartite.

The computation of the number of edges is based on the fact that every edge corresponds to the addition/subtraction of a move $\mathbf{m}$ from the basis $\mathcal{B}_{N}$, therefore we just have to check which moves are admissible for a generic vertex $\mathbf{y}\in\mathcal{F}_{N,t}$.

The set of vertices can be divided into three subsets:
\begin{itemize}
\item the internal vertices, i.e. the vectors with no component equal to 0;
\item the vertices corresponding to vectors with $y_1\neq0$ and $1\leq z<N$ components equal to zero;
\item the vertices corresponding to vectors with $y_1=0$ and $1\leq z^*<N-1$ additional components equal to zero.
\end{itemize}

Now consider the first subset, i.e. the internal vertices. This set has cardinality $\binom{t-1}{N-1}$ and for each vertex in this set every move $\mathbf{m}_K \ 1\leq K\leq N-1$ with every sign $\varepsilon=\pm1$ is admissible. This is a consequence of the absence of entries equal to 0, which means that we can add or subtract 1 from every entry.
Thus each vertex in this set has $2(N-1)$ edges.

Secondly, consider the set of vertices with $z$ zero components and $y_1\neq0$; this set has cardinality given by the number of possible vectors with sum of entries equal to $t$ and $z$ zero components $\binom{t-1}{N-1-z}$ times the possible positions for the $z$ zero components $\binom{N-1}{z}$. For the vertices in this set the $2(N-1-z)$ moves which do not involve the $z$ zero components are admissible and within the ones which involve the zero components only the $z$ moves with $\varepsilon=+1$ are admissible.
Therefore every vertex in this set has $2(N-1-z)+z=2N-2-z$ edges.

Finally, consider the set of vertices with $y_1=0$ and $z^*$ additional null components and denote the total number of zero components $z=z^*+1$. The cardinality of this set is given by the product between the number of possible vectors with sum of entries equal to $t$ and $z$ zero components $\binom{t-1}{N-1-z}$ and the possible positions for the $z^*$ additional zero components $\binom{N-1}{z^*}=\binom{N-1}{z-1}$.
For the vertices in this set $\varepsilon=+1$ is the only admissible sign and if $\varepsilon=+1$ the moves involving the $z^*$ zero components are not admissible; therefore each vertex in this set has $N-1-z^*=N-z$ edges.

Thus the total number of edges is given by the sum of these three terms
\begin{align*}
&2(N-1)\binom{t-1}{N-1}+\\
&\sum_{z=1}^{N-1}(2N-2-z)\binom{t-1}{N-1-z} \binom{N-1}{z}+\\
&\sum_{z=1}^{N-1}(N-z)\binom{t-1}{N-1-z}\binom{N-1}{z-1}
\end{align*}
divided by two (because by counting the edges of each vertex we count the same edge twice), which results in the formula in equation (\ref{eq:edges}).

To prove that $\mathcal{G}$ is bipartite we observe that it is not possible to return to the starting vector by an odd sequence of moves: consider the first component $y_1$ of a generic vector $\mathbf{y}\in\mathcal{F}_{N,t}$ and a generic path of moves.
Every move acts on $y_1$ with a $+1$ or a $-1$. To come back to $y_1$ the sequence of $+1$ and $-1$ has to be even. This proves that $\mathcal{G}$ has no cycle of odd length, hence the graph is bipartite.

\section{Markov basis for $\mathcal{F}_{N,t}^{(\pi)}$}
\label{app:mb}
In this section we give a way to build the Markov basis for $\mathcal{F}_{N,t}^{(\pi)}$ given a positive integer $t$.

First we recall the definition of Markov basis (\cite{sturmfels} and \cite{aokitakemura}):
\begin{definition}
A Markov basis for the $2\times (t+1)$ matrix $A_t^{(\pi)}=\begin{pmatrix}
0&1&\ldots&t\\
1&1&\ldots&1
\end{pmatrix}$ is a finite set of moves $\mathcal{B}_t^{(\pi)}=\lbrace\mathbf{m}_1^{(\pi)},\ldots,\mathbf{m}_K^{(\pi)}\rbrace$ such that
\begin{enumerate}
\item $\mathbf{m}_i^{(\pi)}$ belongs to the integer kernel of $A_t^{(\pi)}$ for every $1\leq i\leq K$, i.e. $A_t^{(\pi)}\mathbf{m}_i^{(\pi)}=\mathbf{0}$ and $\mathbf{m}_i^{(\pi)}\in\mathbb{Z}^{t+1}$ for every $1\leq i\leq K$;
\item for any $\mathbf{f}_1,\mathbf{f}_2$ such that $A_t^{(\pi)}\mathbf{f}_1=A_t^{(\pi)}\mathbf{f}_2=\mathbf{b}$ with $\mathbf{b}\geq 0$, there exist $\Delta>0$, $(\varepsilon_1,\mathbf{m}_{i_{1}}^{(\pi)}),\ldots,(\varepsilon_{\Delta},\mathbf{m}_{i_{\Delta}}^{(\pi)})$ with $\varepsilon_j \in\lbrace-1,+1\rbrace$ and $\mathbf{m}_{i_{j}}^{(\pi)}\in\mathcal{B}_t^{(\pi)}$ for every $j=1,\ldots,\Delta$ satisfying
\begin{gather}
\label{eq:definition}
\mathbf{f}_1=\mathbf{f}_2+\sum_{j=1}^\Delta \varepsilon_j\mathbf{m}_{i_{j}}^{(\pi)}\quad \text{and}\\
\mathbf{f}_2+\sum_{j=1}^\delta \varepsilon_j\mathbf{m}_{i_{j}}^{(\pi)} \geq 0 \text{ for } \delta=1,\ldots,\Delta.
\end{gather}
\end{enumerate}
\end{definition} 

As stated by \cite{sturmfels} and \cite{aokitakemura} the second condition implies that the graph induced over the fiber $\mathcal{F}_{N,t}^{(\pi)}$ by the Markov basis $\mathcal{B}_t^{(\pi)}$ is connected.

In our case we are interested in a Markov Basis for the specific fiber $\mathcal{F}_{N,t}^{(\pi)}=\lbrace \mathbf{f}: A_t^{(\pi)}\mathbf{f}=\begin{pmatrix}
t\\
N
\end{pmatrix} \rbrace$. The following Proposition \ref{prop:basis} constructs a Markov Basis for $\mathcal{F}_{N,t}^{(\pi)}$. With a slight abuse of notation we still denote by $\mathcal{B}_t^{(\pi)}$ such a basis. We denote as $\lfloor x \rfloor$ the floor of $x$, $\lfloor x\rfloor =\max\{m\in \mathbb {Z} \mid m\leq x\}$.

\begin{proposition}
\label{prop:basis}
For any integer $t$ one can build a Markov basis $\mathcal{B}_t^{(\pi)}$ for the fiber $\mathcal{F}_{N,t}^{(\pi)}$ considering the moves $\mathbf{m}_{k,i}$ which are built as follows: for every $2\leq k\leq t$ and for every $1\leq i\leq\lfloor k/2\rfloor$ the $t+1$ vector $\mathbf{m}_{k,i}$ is constructed as follows
\begin{enumerate}
	\item set all the components of $\mathbf{m}_{k,i}$ equal to zero;
	\item set $(\mathbf{m}_{k,i})_0=-1$ and $(\mathbf{m}_{k,i})_k=-1$;
\item set $(\mathbf{m}_{k,i})_i=1$; 
\item set $(\mathbf{m}_{k,i})_{k-i}=(\mathbf{m}_{k,i})_{k-i}+1$.
\end{enumerate}

\begin{proof}
First we observe that for any $\mathbf{m} \in \mathcal{B}_t^{(\pi)}$ its components $(\mathbf{m})_i, i=0,\ldots,t$ are in $\lbrace -1,0,1,2\rbrace$. It follows that $\mathbf{m} \in\mathbb{Z}^{t+1}$. We also observe that $m_i=2$ if and only if $i=k-i$, that is if $i=k/2$.

Secondly, $A_t^{(\pi)}\mathbf{m}=\mathbf{0}$ because
\begin{gather*}
( 0,1,\ldots,t)\mathbf{m}=0\cdot m_0+i\cdot m_i+(k-i)\cdot m_{k-i}+k\cdot m_k=\\
=0\cdot (-1)+i\cdot 1 +(k-i)\cdot 1+k\cdot (-1)=0
\end{gather*}
and 
\begin{gather*}
( 1,1,\ldots,1)\mathbf{m}=1\cdot m_0+1\cdot m_i+1\cdot m_{k-i}+1\cdot m_k=\\
=1\cdot (-1)+1\cdot 1 +1\cdot 1+1\cdot (-1)=0.
\end{gather*}

Thirdly, we prove that the points of $\mathcal{F}_{N,t}^{(\pi)}$ are connected by the moves of $\mathcal{B}_t^{(\pi)}$ by induction over $t$. Preliminarly, we observe that $\mathcal{B}_t^{(\pi)}$ can be considered as the disjoint union of $\mathcal{B}_{t,0}^{(\pi)}$ and  $\mathcal{B}_{t,1}^{(\pi)}$ where $\mathcal{B}_{t,0}^{(\pi)}=\lbrace (x_0,\ldots,x_t) \in \mathcal{B}_t^{(\pi)}: x_t=0 \rbrace$ and $\mathcal{B}_{t,1}^{(\pi)}=\lbrace (x_0,\ldots,x_t)  \in \mathcal{B}_t^{(\pi)}: x_t=-1 \rbrace$. $\mathcal{B}_{t,0}^{(\pi)}$ is obtained with $2 \leq k < t$ and $\mathcal{B}_{t,1}^{(\pi)}$ is obtained with $k=t$.

By construction it holds that
$\mathcal{B}_{t,0}^{(\pi)}=\lbrace (x_0,\ldots,x_{t-1},0): (x_0,\ldots,x_{t-1}) \in \mathcal{B}_{t-1}^{(\pi)}\rbrace$.
  
\begin{itemize}
\item For $t=1$, we have $\mathcal{F}_{N,1}^{(\pi)}=\left\{(N-1,1)\right\}$ and $\mathcal{B}_1^{(\pi)}=\emptyset$.
\item For $t=2$, we have $\mathcal{F}_{N,2}^{(\pi)}=\left\{(N-1,0,1),(N-2,2,0)\right\}$ and $\mathcal{B}_2^{(\pi)}=\left\{(-1,2,-1)\right\}$. It follows that the two points of $\mathcal{F}_{N,2}^{(\pi)}$ are connected by
the move of $\mathcal{B}_2^{(\pi)}$.
\item Let us now suppose that $\mathcal{B}_t^{(\pi)}$ connects $\mathcal{F}_{N,t}^{(\pi)}$ and let us prove that $\mathcal{B}_{t+1}^{(\pi)}$ connects $\mathcal{F}_{N,t+1}^{(\pi)}$. We observe that $\mathcal{F}_{N,t+1}^{(\pi)}$ is the disjoint union of the sets $\tilde{\mathcal{F}}_{N,t+1}^{(\pi)}$ and $\left\{(N-1,0,\ldots,0,1)\right\}$, where $\tilde{\mathcal{F}}_{N,t+1}^{(\pi)}$ contains the points $(x_0,x_1+1,x_2,\ldots,x_t,0)$ with $(x_0,x_1,x_2,\ldots,x_t) \in \mathcal{F}_{N,t}^{(\pi)}$. It is easy to verify that $\mathcal{B}_{t+1,0}^{(\pi)}$ connects all the points of $\tilde{\mathcal{F}}_{N,t+1}^{(\pi)}$ and that $\mathcal{B}_{t+1,1}^{(\pi)}$ connects the point $\left\{(N-1,0,\ldots,0,1)\right\}$ to the points of $\tilde{\mathcal{F}}_{N,t+1}^{(\pi)}$.
\end{itemize}

\end{proof}\qed
\end{proposition}

In general \texttt{4ti2} will provide a larger number of moves than those obtained using Proposition \ref{prop:basis}. This is due to the fact that \texttt{4ti2} gives a basis for \emph{all} the fibers $\lbrace \mathbf{f} :A^{(\pi)}_t \mathbf{f}=\mathbf{b}, \mathbf{b}\geq \mathbf{0}\rbrace$, 
while we have built a Markov basis for the specific fiber $\mathcal{F}^{(\pi)}_{N,t}$, where $\mathbf{b}=\begin{pmatrix}
t\\
N
\end{pmatrix}$.

For $t=6$, \texttt{4ti2} gives 15 moves: the nine listed in equation \eqref{eq:b6} plus the six below
\begin{equation*} \label{eq:b6extra}
\begin{bmatrix*}[r]
0 &	0 &	0 &	0 &	0 &	0 \\
0 &	0 &	0 &	0 &	1 &	1 \\
0 &	0 &	1 &	1 &	-1 &	0 \\
0 &	1 &	-1 &	0 &	0 &	-1 \\
1 &	-1 &	0 &	-2 &	0 &	-1 \\
-2 &	-1 &	-1 &	0 &	-1 &	0 \\
1 &	1 &	1 &	1 &	1 &	1 \\
\end{bmatrix*}.
\end{equation*}
As one can easily check none of these moves is admissible. For example to use the first move $(0,0,0,0,1,-2,1)$ we need a vector $\mathbf{f}$ such that $\sum_{i=0}^6 i f_i \geq 10$, but such $\mathbf{f}$ does not belong to $\mathcal{F}^{(\pi)}_{N,6}$.

\begin{proposition}
The number $K$ of moves in $\mathcal{B}^{(\pi)}_t$ is equal to
\[
K=\begin{cases}
\frac{t^2}{4}\quad \text{if $t$ is even}\\
\frac{t^2-1}{4}\quad \text{if $t$ is odd}\\
\end{cases}.
\]
\begin{proof}
From Proposition \ref{prop:basis} it follows that the total number of moves in $\mathcal{B}^{(\pi)}_t$, for a generic $t$, is given by $\sum_{k=2}^t \lfloor k/2\rfloor$.

Thus if $2\leq k\leq t$ we need to compute the sum of the following sequence
\[
\begin{array}{c|ccccccccc}
k& &2		&3	&4	&5	&6&7&\ldots&t\\
\lfloor k/2\rfloor& &1	&1	&2	&2	&3	&3	&\ldots &\lfloor t/2\rfloor\\
\end{array}.
\]
If $t$ is odd then this sum is
\begin{align*}
\sum_{k=2}^t \lfloor k/2\rfloor&=2\cdot\sum_{k=1}^{(t-1)/2} k=\\
&=2\cdot\frac{1}{2}\cdot\left(\frac{t-1}{2}\cdot\left(\frac{t-1}{2}+1\right)\right)=\\
&=\frac{(t-1)(t+1)}{4}=\frac{t^2-1}{4}.
\end{align*}
If $t$ is even then
\begin{align*}
\sum_{k=2}^t \lfloor k/2\rfloor&=2\cdot\sum_{k=1}^{(t-2)/2} k+\frac{t}{2}=\\
&=2\cdot\frac{1}{2}\cdot\left(\frac{t-2}{2}\cdot\left(\frac{t-2}{2}+1\right)\right)+\frac{t}{2}=\\
&=\frac{t}{2}\left(\frac{t-2}{2}+1\right)=\frac{t^2}{4}.
\end{align*}
\end{proof}\qed
\end{proposition}

\section{Properties of Estimators}
\label{app:estimator}
In this appendix we show the proofs of the 2 properties of the estimators $\mathbb{I}_{(U(\mathbf{y})\leq u)}$ and $F_U(u\vert \pi)$ presented in Section \ref{sec:compare}.
First we prove the relation between the variances of the two estimators, then, thanks to Lemma \ref{lemma}, we prove a similar result for the mean absolute deviation.
\paragraph{\textbf{Proposition 2}}

$\variance_{p}\left[\mathbb{I}_{(U(\mathbf{y})\leq u)}\right] \geq \variance_{p_\pi} \left[F_U(u\vert \pi)\right]$.
\begin{proof}
From Proposition \ref{prop} both $\mathbb{I}_{(U(\mathbf{y})\leq u)}$ and $F_U(u\vert \pi)$ are unbiased estimator of the distribution of $U$ over the fiber $\fiber$, $F_U(u\vert \fiber)$.
Then it is enough to show that 
\[
\mathbb{E}_{p}\left[(\mathbb{I}_{(U(\mathbf{y})\leq u)})^2\right] \geq \mathbb{E}_{p_\pi}\left[(F_U(u\vert \pi)^2)\right]
\]
From  $(\mathbb{I}_{(U(\mathbf{y})\leq u)})^2=\mathbb{I}_{(U(\mathbf{y})\leq u)}$  we have
\begin{align*}
\mathbb{E}_{p}\left[(\mathbb{I}_{(U(\mathbf{y})\leq u)})^2\right]&=\mathbb{E}_{p}\left[\mathbb{I}_{(U(\mathbf{y})\leq u)}\right] \\
&=\sum_{\mathbf{y}\in\fiber} p(\mathbf{y})\mathbb{I}_{(U(\mathbf{y})\leq u)}\\
&=\sum_{\pi\subseteq\mathcal{F}_{N,t}} \sum_{\mathbf{y}\in \pi} p(\mathbf{y}) \mathbb{I}_{(U(\mathbf{y})\leq u)} \\
&=\sum_{\pi\subseteq\mathcal{F}_{N,t}} p(\mathbf{y}_\pi)\sum_{\mathbf{y}\in \pi}\mathbb{I}_{(U(\mathbf{y})\leq u)} \\
&=\sum_{\pi\subseteq\mathcal{F}_{N,t}} \#\pi p(\mathbf{y}_\pi)\sum_{\mathbf{y}\in \pi}\frac{\mathbb{I}_{(U(\mathbf{y})\leq u)}}{\#\pi} \\
&=\sum_{\pi\subseteq\mathcal{F}_{N,t}} p_\pi(\pi)\sum_{\mathbf{y}\in \pi}\frac{\mathbb{I}_{(U(\mathbf{y})\leq u)}}{\#\pi} \\
&\geq \sum_{\pi\subseteq\mathcal{F}_{N,t}} p_\pi(\pi)\left(\sum_{\mathbf{y}\in \pi}\frac{\mathbb{I}_{(U(\mathbf{y})\leq u)}}{\#\pi}\right)^2 \\
&=\mathbb{E}_{p_\pi}\left[(F_U(u\vert \pi)^2)\right]
\end{align*}
where $\mathbf{y}_\pi$ is any of the element of the orbit $\pi$. The $\geq$ sign comes from 
\begin{align*}
&0\leq \left(\sum_{\mathbf{y}\in \pi}\frac{\mathbb{I}_{(U(\mathbf{y})\leq u)}}{\#\pi}\right) \leq 1 \\
&\Rightarrow \left(\sum_{\mathbf{y}\in \pi}\frac{\mathbb{I}_{(U(\mathbf{y})\leq u)}}{\#\pi}\right)^2 \leq
\left(\sum_{\mathbf{y}\in \pi}\frac{\mathbb{I}_{(U(\mathbf{y})\leq u)}}{\#\pi}\right).
\end{align*}
\end{proof}\qed

To prove the result in Proposition \ref{prop3} for the mean absolute deviation (MAD), we need the following Lemma.
\begin{lemma}
\label{lemma}
Let $x,y\in\left[0,1\right]$, then the following facts hold:
\begin{enumerate}
\item $x-2xy+y\geq0$;
\item $\vert x-y\vert\leq x-2xy+y$.
\end{enumerate}
\begin{proof}
\begin{enumerate}
\item Having $x,y\in\left[0,1\right]$ implies $x\geq x^2$ and $y\geq y^2$. Thus
\[x-2xy+y\geq x^2-2xy+y^2=(x-y)^2\geq0.\]
\item We denote $M(x,y)=x-2xy+y\geq0$. Then $\vert x-y\vert\leq x-2xy+y$ is equivalent to
\[\begin{cases}
x-y\leq M(x,y)\\
x-y\geq -M(x,y)
\end{cases}.\]

Consider the first inequality:
\[
x-y\leq x-2xy+y\quad \Rightarrow \quad 2xy-2y\leq0.
\]
If $y=0$ this is clearly true, while if $y\neq0$ we have $x-1\leq0$ which is always true given that $x,y\in\left[0,1\right]$.

By a similar computation one can verify that the second inequality always holds true.

Hence we can conclude $\vert x-y\vert\leq x-2xy+y$ for every $x,y\in\left[0,1\right]$.
\end{enumerate}
\end{proof}\qed
\end{lemma}

\paragraph{\textbf{Proposition 3}}
$\mad_{p}\left[\mathbb{I}_{(U(\mathbf{y})\leq u)}\right]\geq\mad_{p_\pi}\left[F_U(u\vert\pi)\right]$.
\begin{proof}
\begin{align*}
&\mad_{p}\left[\mathbb{I}_{(U(\mathbf{y})\leq u)}\right]=\mathbb{E}\left[\vert\mathbb{I}_{(U(\mathbf{y})\leq u)}-F_U(u\vert\fiber) \vert\right]\\
&=\sum_{\mathbf{y}\in\fiber} p(\mathbf{y})\vert\mathbb{I}_{(U(\mathbf{y})\leq u)}-F_U(u\vert\fiber) \vert\\
&=\sum_{\pi\subseteq\mathcal{F}_{N,t}} \sum_{\mathbf{y}\in \pi} p(\mathbf{y}) \vert\mathbb{I}_{(U(\mathbf{y})\leq u)}-F_U(u\vert\fiber) \vert \\
&=\sum_{\pi\subseteq\mathcal{F}_{N,t}}p_{\pi}(\pi)\cdot \frac{1}{\#\pi }  \sum_{\mathbf{y}\in \pi} \vert\mathbb{I}_{(U(\mathbf{y})\leq u)}-F_U(u\vert\fiber) \vert.
\end{align*}
We divide the vectors $\mathbf{y}\in\pi$ into two classes $C_0$ and $C_1=\bar{C_0}$, such that $C_0=\left\lbrace\mathbf{y}\in \pi:\mathbb{I}_{(U(\mathbf{y})\leq u)}=0\right\rbrace$ and $C_1=\left\lbrace\mathbf{y}\in \pi:\mathbb{I}_{(U(\mathbf{y})\leq u)}=1\right\rbrace$, then 
\begin{align*}
&\mad_{p}\left[\mathbb{I}_{(U(\mathbf{y})\leq u)}\right]=\\
&=\sum_{\pi\subseteq\mathcal{F}_{N,t}}p_{\pi}(\pi)\cdot \frac{1}{\#\pi }  \left(\sum_{C_0} F_U(u\vert\fiber)+\sum_{C_1} \left(1-F_U(u\vert\fiber)\right)\right)\\
&=\sum_{\pi\subseteq\mathcal{F}_{N,t}}p_{\pi}(\pi)\cdot \frac{1}{\#\pi }  \left(\#C_0\cdot F_U(u\vert\fiber)+\#C_1\cdot \left(1-F_U(u\vert\fiber)\right)\right).
\end{align*}
By looking at the definition of $F_U(u\vert\pi)$ in equation \eqref{eq:dist_orbit} we observe that $\#C_1/\#\pi=F_U(u\vert\pi)$ and $\#C_0/\#\pi =1-\#C_1/\#\pi=1-F_U(u\vert\pi)$, thus
\begin{align*}
&\mad_{p}\left[\mathbb{I}_{(U(\mathbf{y})\leq u)}\right]=\\
&=\sum_{\pi\subseteq\mathcal{F}_{N,t}}p_{\pi}(\pi) \left(F_U(u\vert\fiber)-2F_U(u\vert\pi)F_U(u\vert\fiber)+F_U(u\vert\pi)\right).
\end{align*}
Now, by Lemma \ref{lemma}, it holds that $\vert F_U(u\vert\fiber)-F_U(u\vert\pi)\vert \leq M(\fiber,\pi)$, where $M(\fiber,\pi)=F_U(u\vert\fiber)-2F_U(u\vert\pi)F_U(u\vert\fiber)+F_U(u\vert\pi)$. Therefore
\begin{align*}
&\mad_{p}\left[\mathbb{I}_{(U(\mathbf{y})\leq u)}\right]=\\
&=\sum_{\pi\subseteq\mathcal{F}_{N,t}}p_{\pi}(\pi) \left(F_U(u\vert\fiber)-2F_U(u\vert\pi)F_U(u\vert\fiber)+F_U(u\vert\pi)\right)\\
&\geq \sum_{\pi\subseteq\mathcal{F}_{N,t}}p_{\pi}(\pi) \vert F_U(u\vert\fiber)-F_U(u\vert\pi)\vert\\
&=\mad_{p_\pi}\left[F_U(u\vert\pi)\right].
\end{align*}
\end{proof}\qed

\nocite{*}
\bibliographystyle{spbasic}      

\end{document}